\documentclass[11pt]{article}
\usepackage[top=27mm,bottom=27mm,left=22mm,right=22mm]{geometry}
\usepackage[T1]{fontenc}
\usepackage[utf8]{inputenc}
\usepackage{lmodern}
\usepackage{amsmath,amssymb,amsthm}
\usepackage{mathtools}
\usepackage{microtype}
\usepackage{enumitem}
\usepackage{booktabs}
\usepackage{needspace}
\usepackage{hyperref}
\usepackage[nameinlink,noabbrev]{cleveref}

\numberwithin{equation}{section}
\theoremstyle{plain}
\newtheorem{theorem}{Theorem}[section]
\newtheorem{lemma}[theorem]{Lemma}
\newtheorem{proposition}[theorem]{Proposition}
\newtheorem{corollary}[theorem]{Corollary}
\newtheorem{conjecture}[theorem]{Conjecture}
\theoremstyle{definition}
\newtheorem{definition}[theorem]{Definition}
\theoremstyle{remark}

\newcommand{\F}{\mathbb F}
\newcommand{\A}{\mathcal A}
\newcommand{\E}{\mathcal E}
\newcommand{\C}{\mathcal C}
\newcommand{\PG}{\operatorname{PG}}
\newcommand{\RS}{\operatorname{RS}}
\newcommand{\ev}{\operatorname{ev}}
\newcommand{\Disc}{\operatorname{Disc}}
\newcommand{\rank}{\operatorname{rank}}
\newcommand{\dom}{\operatorname{dom}}
\newcommand{\ceil}[1]{\left\lceil #1\right\rceil}
\hypersetup{
 colorlinks=true,
 linkcolor=blue,
 citecolor=blue,
 urlcolor=blue,
 filecolor=blue,
 pdftitle={The Cheng--Murray Conjecture on Deep Holes of Full-Length Reed--Solomon Codes},
 pdfauthor={Weijun Fang}
}

\title{The Cheng--Murray Conjecture on Deep Holes of Full-Length Reed--Solomon Codes}
\author{Weijun Fang\\
\small State Key Laboratory of Cryptography and Digital Economy Security\\
\small Key Laboratory of Cryptologic Technology and Information Security, Ministry of Education\\
\small School of Cyber Science and Technology, Shandong University\\
\small Qingdao 266237, China\\
\small \texttt{fwj@sdu.edu.cn}}
\date{}
\begin{document}
\maketitle

\begin{abstract}
Let $\RS_q(k)$ denote the full-length Reed--Solomon code
over $\F_q$. Every word $u=(u_x)_{x\in\F_q} \in\F_q^q$ has a unique
interpolation polynomial $U(X)\in\F_q[X]$ of degree less than
$q$, satisfying $U(x)=u_x$ for all $x\in\F_q$.
It is well known that every word with $\deg U=k$ is a
deep hole. The Cheng--Murray conjecture asserts that these
are the only deep holes for $3\le k\le q-2$, except when
$q$ is even and $k=q-3$. Previous results established the
conjecture for $k\le p-1$ and for
$k\ge\lfloor(q-1)/2\rfloor$, where $p$ is the characteristic
of $\F_q$.
In this paper, we prove the Cheng--Murray conjecture in the
remaining range $3\le k\le\lfloor q/2\rfloor$, thereby
resolving it completely.
Our proof combines two complementary approaches.
The tangent envelope of a marked arc yields differential
constraints that establish the conjecture for
$3\le k\le q/2$ when $q$ is even, and for
$3\le k\le\lceil q/3\rceil$ when $q$ is odd.
For the remaining cases with $q$ odd, we combine
local syndrome compatibility with combinatorial
arguments to rule out additional deep holes.
\end{abstract}



\section{Introduction}
Let $\F_q$ denote the finite field of order $q$, where $q$ is a prime power.
For a word $u\in\F_q^n$ and a linear code
$C\subseteq\F_q^n$, the distance from $u$ to $C$ is
defined by
\[
d(u,C)=\min_{c\in C}d_H(u,c),
\]
where the Hamming distance $d_H$ counts the coordinates in which
its two arguments differ.
The \emph{covering radius} of $C$ is
\[
\rho(C)=\max_{u\in\F_q^n}d(u,C).
\]
A word $u$ satisfying $d(u,C)=\rho(C)$ is called
a \emph{deep hole} of $C$. Thus deep holes are the words farthest from $C$.

For $1\le k<q$, the
\emph{full-length Reed--Solomon code} over $\F_q$ is
defined as
\[
\RS_q(k)=
\{(f(x))_{x\in\F_q}:
f\in\F_q[X],\ \deg f<k\}.
\]
Every word $u=(u_x)_{x\in\F_q}\in\F_q^q$
has a unique interpolation polynomial
$U\in\F_q[X]$ of degree less than $q$ such that
$U(x)=u_x$ for all $x\in\F_q$.
We call $U$ the \emph{canonical representative} of $u$.
For $k\le\deg U<q$, we have
\[
q-\deg U
\le d(u,\RS_q(k))
\le q-k.
\]
The lower bound follows from the root bound for $U-f$,
and the upper bound from interpolation on any $k$ coordinates.
Consequently, $\rho(\RS_q(k))=q-k$, and
every word with $\deg U=k$ is a deep hole.
These are called the \emph{expected deep holes}; all other deep holes
are called \emph{nonexpected}.

In 2007, Cheng and Murray~\cite{CM} conjectured
that these are the only deep holes of
$\RS_q(k)$. We state the conjecture in its
corrected form, excluding the known exception
in even characteristic~\cite{Kaipa,ZhangZhouSurvey}.

\begin{conjecture}[Cheng--Murray]\label{conj:CM}
Let $3\le k\le q-2$, excluding the case where $q$ is even and $k=q-3$. A word $u\in\F_q^q$ is a deep hole of
	$\RS_q(k)$ if and only if its canonical representative
	has degree $k$.
\end{conjecture}

Equivalently, under the hypotheses of Conjecture~\ref{conj:CM},
every polynomial $U$ with $k<\deg U<q$ agrees with some
polynomial of degree less than $k$ at $k+1$ distinct field elements.
In the excluded even-characteristic case $k=q-3$, the additional
deep holes are represented by $aX^{q-2}+f(X)$, where
$a\in\F_q^\times$ and $\deg f<q-3$~\cite[Theorem~1]{Kaipa}.

For $k=2$, the classification holds when $q$ is odd by
\cite[Theorem~10(1)]{ZCL}, but fails in even characteristic;
see the example following Lemma~\ref{lem:arc}.

\subsection{Related Work}

Much of the early work on Reed--Solomon deep holes focused on
words whose interpolation polynomials have degree $k+d$,
where $d\ge1$.
Cheng and Murray~\cite{CM} reduced the exclusion of such
deep holes to the existence of rational points with distinct
coordinates on certain algebraic hypersurfaces.
Subsequent results based on character sums and
distinct-coordinate counting were obtained by
Li and Wan~\cite{LiWan}, Liao~\cite{Liao},
and Zhu and Wan~\cite{ZhuWan}.
Cafure, Matera, and Privitelli~\cite{CMP} analyzed the
singularities of related symmetric hypersurfaces to derive
further deep-hole exclusions.
In the reverse direction, Zhang and Wan~\cite{ZhangWanSym}
used connections with Reed--Solomon deep holes and finite
geometry to obtain lower bounds for rational points on
complete symmetric hypersurfaces.
These approaches give deep-hole exclusions under various
conditions on the field size, code dimension, and
interpolation degree.

The first two excess degrees have received particular
attention.
For $d=1$, Li and Wan~\cite{LiWanSubset} reduced the
deep-hole question to a subset-sum problem and obtained
explicit counting formulas.
For full-length codes with $2\le k\le q-2$, the
degree-$(k+1)$ deep holes are completely classified:
none occur in odd characteristic, while in even
characteristic they occur precisely when $k=q-3$ and
the coefficient of $X^k$ vanishes
\cite[Theorem~1]{LiZhu}.
For $d=2$, Zhang, Fu, and Liao~\cite{ZFL} constructed
deep holes for proper evaluation sets and excluded
degree-$(k+2)$ deep holes of full-length codes in
odd characteristic when $q>5$ and $2\le k\le q-3$.
Li and Zhu~\cite{LiZhu} subsequently determined
the exact error distances of degree-$(k+1)$ words
and obtained exact distances or upper bounds for
degree-$(k+2)$ words in various parameter ranges.
In a related direction, Li and Wan~\cite{LiWanDist}
studied the distance distribution of Reed--Solomon
codes, while Gao and Li~\cite{GaoLi} improved the
corresponding counting estimates and derived further
degree-dependent deep-hole exclusions.
These degree-restricted results do not, by themselves,
yield a classification for arbitrary interpolation degrees.

A complementary approach relates deep holes to
one-coordinate extensions of maximum-distance-separable
(MDS) codes.
Seroussi and Roth~\cite{SR} established structural
results on MDS extensions of generalized Reed--Solomon
codes.
Zhuang, Cheng, and Li~\cite[Theorem~10]{ZCL}
used finite geometry to prove the full-length
classification in odd characteristic $p$ when
$k+1\le p$ or
$3\le q-p+1\le k+1\le q-2$.
In particular, the Cheng--Murray conjecture holds
over odd prime fields.
Kaipa~\cite{Kaipa} exploited the correspondence
between deep-hole syndromes and one-coordinate
MDS extensions to obtain the classification for
$k\ge\lfloor(q-1)/2\rfloor$, including the known
even-characteristic exception.
For a comprehensive survey of the deep-hole problem
for generalized Reed--Solomon codes, including its
connections with finite geometry and related
computational questions, we refer to
Zhang and Zhou~\cite{ZhangZhouSurvey}.

\subsection{Main Result and Proof Outline}

The existing results do not settle the conjecture for all
$3\le k\le\lfloor q/2\rfloor$ over extension fields.
Our main theorem fills this gap.

\begin{theorem}\label{thm:main}
	Let $q$ be a prime power and $3\le k\le\lfloor q/2\rfloor$.
	A word $u\in\F_q^q$ is a deep hole of $\RS_q(k)$ if and
	only if its canonical representative $U$ has degree $k$.
\end{theorem}

Together with the known classification for
$k\ge\lfloor(q-1)/2\rfloor$,
Theorem~\ref{thm:main} completes the proof of
Conjecture~\ref{conj:CM}.

Our proof combines two complementary forms of rigidity
arising from the deep-hole condition: the geometry of
marked arcs and the compatibility of local syndromes.

We begin with the geometric approach.
Let $U$ be the canonical representative of a deep hole,
and associate with it the marked set
\[
\A_U=
\{[1:x:\cdots:x^{k-1}:U(x)]:x\in\F_q\}
\cup\{[0:\cdots:0:1]\}
\subseteq\PG(k,q).
\]
Here $\PG(k,q)$ is the $k$-dimensional projective space over $\F_q$,
and the last point is the marked point.
The deep-hole condition is equivalent to $\A_U$ being a
$(q+1)$-arc, meaning that every $k+1$ of its points are
projectively independent.
Projection from the marked point gives the affine part
of a normal rational curve.
Exploiting this prescribed projection and the
tangent-envelope theorem in the form given by
Ball and Lavrauw~\cite{BLtensor,BLsurvey}, we explicitly determine
the restriction of the tangent-envelope form to the
dual hyperplane parametrizing hyperplanes through
the marked point.
The resulting boundary form is expressed in terms
of the discriminant in odd characteristic and its
polynomial square root in characteristic two, together
with the appropriate leading-coefficient factors.
We then compute the first transverse coefficient
of the envelope and derive linear differential
constraints on $U$.
These constraints establish the conjecture for
$3\le k\le q/2$ in even characteristic and
$q>3(k-1)$ in odd characteristic, using the known
error-distance result for degree-$(k+1)$ words
to eliminate the residual possibility in
characteristic three.

For the remaining range, which occurs only in
odd characteristic, we develop a local-syndrome
approach.
We show that every $(k+2)$-subset of $\F_q$
admits a unique template agreeing with $U$,
either a degree-$k$ polynomial or a rational function
$f(X)+c/(X-\zeta)$ with $\deg f<k$, $c\in\F_q^\times$,
and $\zeta\in\F_q$.
A failure to extend such a template to an additional
point produces a quadratic local syndrome outside
the conic of evaluation functionals.
The associated obstruction graph is a disjoint union of edges
and isolated vertices in which the original $(k+2)$-subset
is independent.
The key step is a compatibility argument showing
that distinct failed extensions of a fixed subset
produce distinct projective quadratic syndromes.
Combining this noncollision property with a lower
bound on the density of bad subsets forces
a quantitative lower bound on the number of
available quadratic obstructions.
On the other hand, the matching structure permits
an exact count of the corresponding independent
subsets.
Comparing these bounds excludes nonexpected deep
holes in the remaining range for all odd $q\ge41$.
The remaining odd proper extension fields
$q\in\{9,25,27\}$ are handled by elementary
matching and counting arguments, together with
a restricted sum--product lemma over $\F_{27}$.

\Needspace{12\baselineskip}
\noindent\emph{Recent independent work.}
Around the time the present manuscript was completed and
posted on arXiv~\cite{FangArxiv}, two independently developed
preprints on the Cheng--Murray conjecture also appeared.

Chen and Huang~\cite{Chen-Huang} establish the deep-hole
classification for every prime power $q>19683$ and
$3\le k\le q-2$, including the known even-characteristic
exception. They also parametrize the dimension-two deep holes
in even characteristic by normalized $o$-polynomials. Their proof
combines planar-arc completion, quadratic quotient
algebras, and weighted independent-set estimates for
involution graphs. Although both approaches use finite
geometry and syndrome methods, our proof proceeds
differently: we derive differential constraints from the
tangent envelope of a marked arc, and use local template
compatibility and exact matching counts for quadratic
obstructions to handle the remaining cases.

Fan~\cite{Fan} approaches the problem by prescribing the
interpolation degree $\deg U=k+m$, where $m\ge3$.
For $3\le m\le7$, he classifies all such deep holes over
every finite field, including the exceptional dimension-two
families in even characteristic. For arbitrary $m\ge3$,
he excludes degree-$(k+m)$ deep holes when $k\ge3$ and
$\sqrt q\ge12(m+1)\log(2m+2)$, and obtains a complete
classification when $q>16(m+2)^4$. His methods rely on
complete homogeneous symmetric functions, agreement
estimates, and finite-field arguments, supplemented by
certified computations for certain small cases.

Thus, Chen and Huang establish the classification
for $q>19683$, while Fan's results either impose
field-size conditions or concern specified excess degrees.
In contrast, Theorem~\ref{thm:main} establishes
the classification for every prime power $q$ and
every $3\le k\le\lfloor q/2\rfloor$, without any
field-size restriction or restriction on the
interpolation degree.
Combined with Kaipa's high-dimensional result
~\cite{Kaipa}, it completely resolves
Conjecture~\ref{conj:CM}.

\subsection{Organization}

Section~\ref{sec:interpolation} characterizes deep holes by marked arcs.
Section~\ref{sec:envelope} determines the boundary form and first
transverse coefficient of the tangent envelope.
Section~\ref{sec:differential} proves the classification in even
characteristic and for odd $q>3(k-1)$.
Section~\ref{sec:templates} establishes the existence and uniqueness
of local templates, while Section~\ref{sec:quadratic} uses quadratic
obstructions to settle odd $q\ge41$.
Section~\ref{sec:completion} completes the remaining small-field
cases and proves the main theorem. The uniform quadratic-matching
estimate is proved in Appendix~\ref{app:quadratic}.

\section{A Marked-Arc Characterization of Deep Holes}\label{sec:interpolation}
In this section, we characterize deep holes in terms of marked
projective arcs and describe their associated tangent hyperplanes.
For a nonnegative integer $m$, let $\F_q[X]_{\le m}$
denote the space of polynomials over $\F_q$ of degree
at most $m$. We write $[X^j]H$ for the coefficient of
$X^j$ in $H$, and $\binom{S}{r}$ for the family of
$r$-subsets of a finite set $S$.
For $F\in\F_q[X]$, write
\[\ev(F)=(F(x))_{x\in\F_q}.\]

For any $S\subseteq\F_q$, put
\[
P_S(X)=\prod_{a\in S}(X-a),\qquad P_{\varnothing}(X)=1.
\]
For nonempty $S$ and $\deg H<|S|$, Lagrange interpolation takes the form
\[
H(X)=\sum_{a\in S}H(a)\frac{P_S(X)}{(X-a)P_S'(a)}.
\]
Comparing the coefficients of $X^{|S|-1}$, we obtain
\begin{equation}\label{eq:lagrange}
\sum_{a\in S}\frac{H(a)}{P_S'(a)}
=[X^{|S|-1}]H
\qquad(\deg H<|S|).
\end{equation}

Since $\rho(\RS_q(k))=q-k$, a vector with canonical representative $U$ is a deep hole precisely when no polynomial of degree less than $k$ agrees with $U$ at $k+1$ distinct elements of $\F_q$. Equivalently,
\begin{equation}\label{eq:deep-hole-criterion}
\bigl|\{x\in\F_q:U(x)=f(x)\}\bigr|\le k
\qquad\text{for every }f\in\F_q[X],\ \deg f<k.
\end{equation}

Associate with $U$ the set
\[
\A_U=\{[1:x:\cdots:x^{k-1}:U(x)]:x\in\F_q\}\cup\{P_\infty\}
\subseteq\PG(k,q),
\]
where
\[
P_\infty=[0:\cdots:0:1]
\]
is the \emph{marked point}. Recall that a set in $\PG(k,q)$ is an arc if every $k+1$ of its points are projectively independent.

A hyperplane in $\PG(k,q)$ has an equation
\[
\sum_{i=0}^{k-1}a_iY_i+zY_k=0.
\]
Identify its dual coordinates with $(g,z)$, up to simultaneous
nonzero scaling, where $g(X)=\sum_{i=0}^{k-1}a_iX^i$.
It contains the finite point $[1:x:\cdots:x^{k-1}:U(x)]$
if and only if $g(x)+zU(x)=0$, and contains $P_\infty$
if and only if $z=0$.

When evaluating a form in the coefficients of a polynomial, we
identify that polynomial with its coefficient vector. Thus, for a
form $F$ in $a_0,\ldots,a_r$ and $g(X)=\sum_{i=0}^r a_iX^i$,
we write $F(g)=F(a_0,\ldots,a_r)$. For a form $\Phi$
in $a_0,\ldots,a_{k-1},z$, we similarly write
$\Phi(g,z)=\Phi(a_0,\ldots,a_{k-1},z)$.
This convention will also be used for the discriminant forms
and their coefficient derivatives.

The following connection between deep holes and arcs is standard; see, for example, \cite{ZCL,Kaipa}.

\begin{lemma}\label{lem:arc}
A vector $u\in\F_q^q$ with canonical representative $U$ is a deep hole of $\RS_q(k)$ if and only if $\A_U$ is a $(q+1)$-arc in $\PG(k,q)$.
\end{lemma}

\begin{proof}
A set of $k+1$ finite points of $\A_U$ is contained in a hyperplane precisely when
\[
g(x)+zU(x)=0
\]
holds at their $k+1$ distinct parameters for some $\deg g\le k-1$ and $(g,z)\ne(0,0)$. Necessarily $z\ne0$, since otherwise a nonzero polynomial of degree at most $k-1$ would have $k+1$ roots. Thus $U$ agrees at these points with the polynomial $-g/z$ of degree less than $k$, and \eqref{eq:deep-hole-criterion} establishes the equivalence for finite subsets. If one of the $k+1$ points is $P_\infty$, expansion along its coordinate leaves a Vandermonde determinant on the remaining $k$ distinct field elements, which is nonzero.
\end{proof}

The original Cheng--Murray conjecture also included $k=2$, but its
predicted degree-two classification is false in even characteristic.
Indeed, let $q=2^m$, $m\ge3$, and set $U(X)=X^{2^{m-1}}$. For fixed $x\in\F_q$ and $y\ne x$,
\[
 \frac{U(y)-U(x)}{y-x}=(y-x)^{2^{m-1}-1}.
\]
Since $\gcd(2^{m-1}-1,2^m-1)=1$, these $q-1$ slopes are
distinct. Thus no three finite points of $\A_U$ are collinear;
moreover, any line through $P_\infty$ contains at most one
finite point of $\A_U$. By Lemma~\ref{lem:arc}, $\ev(U)$ is
a deep hole of $\RS_q(2)$ although $\deg U>2$.

Henceforth suppose $k\ge3$ and that $U$ represents a deep hole.
Let
$S\subseteq\A_U$ have size $k-1$. Its points are projectively
independent and span a $(k-2)$-dimensional subspace
$\langle S\rangle\subseteq\PG(k,q)$.
The hyperplanes containing $\langle S\rangle$ form a pencil,
naturally identified with a projective line in the dual space
$\PG(k,q)^\vee$.

For each point $P\in\A_U\setminus S$, the span $\langle S\cup\{P\}\rangle$ is a hyperplane belonging to this pencil. Moreover, distinct points of $\A_U\setminus S$ determine distinct hyperplanes: otherwise a single hyperplane would contain $S$ together with two further points of $\A_U$, hence $k+1$ points of the arc, a contradiction. Thus exactly
\[
|\A_U\setminus S|=(q+1)-(k-1)=q-k+2
\]
members of the pencil contain an additional point of $\A_U$.

Since a pencil contains $q+1$ hyperplanes, the remaining
\[
(q+1)-(q-k+2)=k-1
\]
members meet $\A_U$ precisely in the set $S$. We call these the \emph{tangent hyperplanes determined by $S$}.

\section{Tangent-Envelope Identities}\label{sec:envelope}
In this section, we determine the boundary form and the first
transverse coefficient of the tangent envelope of $\A_U$.
These identities provide the basis for the differential argument
in Section~\ref{sec:differential}.

\subsection{The Boundary Form of the Tangent Envelope}

Put
\[
\epsilon=
\begin{cases}
	1, & q \text{ even},\\
	2, & q \text{ odd}.
\end{cases}
\]
We use the following specialization of
\cite[Theorem~7]{BLtensor}; see also
\cite[Theorem~39]{BLsurvey}. This follows by replacing the dimension parameter
in those results by $k+1$ and taking $t=k-1$
and $m=\epsilon$.

\begin{theorem}[Tangent-envelope theorem]\label{thm:envelope}
    Let $q$ be a prime power, $3\le k\le q-2$, and let
    $U\in\F_q[X]$, $\deg U<q$, be the canonical representative
    of a deep hole of $\RS_q(k)$. Assume
    $q\ge(\epsilon+1)(k-1)$.
	Then there exists a nonzero homogeneous form $\Phi$ of degree $\epsilon(k-1)$ on the dual projective space $\PG(k,q)^\vee$ with the following property.

	For every $(k-1)$-subset $S\subseteq\A_U$, let $\ell_S\subseteq\PG(k,q)^\vee$ denote the projective line parametrizing the hyperplanes containing $\langle S\rangle$, and let
	\[
	\tau_1,\ldots,\tau_{k-1}\in\ell_S
	\]
	be the points corresponding to the tangent hyperplanes determined by $S$. Then
	\begin{equation}\label{eq:envelope-restriction}
		\left.\Phi\right|_{\ell_S}
		=
		c_S\prod_{i=1}^{k-1}L_{\tau_i}^{\,\epsilon}
	\end{equation}
	for some $c_S\in\F_q^\times$, where $L_{\tau_i}$ is a linear form on $\ell_S$ vanishing at $\tau_i$.
\end{theorem}

We refer to $\Phi$ as a tangent-envelope form of $\A_U$.

Recall the dual coordinates introduced in
Section~\ref{sec:interpolation}. Here $g(X)=\sum_{i=0}^{k-1}a_iX^i$,
and $(g,z)$ represents the hyperplane
$\sum_{i=0}^{k-1}a_iY_i+zY_k=0$, up to simultaneous nonzero
scaling. Thus $\Phi(g,z)$ is homogeneous in
$a_0,\ldots,a_{k-1},z$. The identity
\eqref{eq:envelope-restriction} is one of homogeneous polynomials;
in particular, the tangent points are precisely the zeros of
$\Phi|_{\ell_S}$, each of multiplicity $\epsilon$.

The dual hyperplane $z=0$ parametrizes the hyperplanes through
the marked point $P_\infty$. Since the projection of $\A_U$
from $P_\infty$ is prescribed, we can determine the boundary form
$\Phi(g,0)$ explicitly.

\begin{definition}[Discriminant and its square root]\label{def:discriminant}
For $r\ge2$ and $g(X)=\sum_{i=0}^r a_iX^i$ with
$a_r\ne0$, write $g(X)=a_r\prod_{i=1}^r(X-x_i)$
over a splitting field. The \emph{discriminant form} is
\begin{equation}
 \Disc_r(g)=a_r^{2r-2}\prod_{1\le i<j\le r}(x_i-x_j)^2.
\end{equation}
In characteristic two, define its polynomial square root by
\begin{equation}
 \delta_r(g)=a_r^{r-1}\prod_{1\le i<j\le r}(x_i+x_j),
 \qquad \Disc_r(g)=\delta_r(g)^2.
\end{equation}
These root expressions define polynomials in the coefficients
and extend to $a_r=0$. They are homogeneous of coefficient
degrees $2r-2$ and $r-1$, respectively. Set
$\Disc_1=\delta_1=1$.
\end{definition}

The following interpolation lemma allows us to extend identities
established for the polynomials $P_T$ to arbitrary polynomials of
degree at most $k-1$.

\begin{lemma}\label{lem:coefficient-interpolation}
	Let $e,r\ge0$ and $q\ge e+r$. Fix a subset $E\subseteq\F_q$ of size $e+r$. A homogeneous form of degree $e$ in the $r+1$ coefficients of
	\[
	g(X)=a_0+a_1X+\cdots+a_rX^r
	\]
	is uniquely determined by its values at the polynomials
	\[
	P_T(X)=\prod_{x\in T}(X-x),
	\qquad
	T\in\binom{E}{r}.
	\]
	Equivalently, if a homogeneous degree-$e$ form
	$F(a_0,\ldots,a_r)$ vanishes at every $P_T$, then $F\equiv0$.
\end{lemma}

\begin{proof}
	For $T\in\binom{E}{r}$, define
	\[
	L_T(g)
	=
	\prod_{x\in E\setminus T}
	\frac{g(x)}{P_T(x)}.
	\]
	Since $|E\setminus T|=e$, the polynomial $L_T$ is a homogeneous form of degree $e$ in $a_0,\ldots,a_r$. Moreover,
	\[
	L_T(P_{T'})
	=
	\begin{cases}
		1,&T'=T,\\
		0,&T'\ne T.
	\end{cases}
	\]
	Indeed, if $T'\ne T$, then some $x\in T'\setminus T$ belongs to $E\setminus T$, and the corresponding factor vanishes.

	Thus the forms $L_T$ are linearly independent. Their number is $\binom{e+r}{r},$
	which is the dimension of the space of homogeneous degree-$e$ forms in $r+1$ variables. Hence they form a basis, and the assertion follows.
\end{proof}

\begin{proposition}\label{prop:boundary}
    Let $q$ be a prime power, $3\le k\le q-2$, and let
    $U\in\F_q[X]$, $\deg U<q$, be the canonical representative
    of a deep hole of $\RS_q(k)$. Assume
    $q\ge(\epsilon+1)(k-1)$.
	The envelope $\Phi$ supplied by Theorem~\ref{thm:envelope} may be scaled so that, for every
	\[
	g(X)=\sum_{i=0}^{k-1}a_iX^i,
	\]
	one has
	\begin{equation}\label{eq:boundary}
		\Phi(g,0)
		=
		\begin{cases}
			a_{k-1}\delta_{k-1}(g),&q\text{ even},\\
			a_{k-1}^2\Disc_{k-1}(g),&q\text{ odd}.
		\end{cases}
	\end{equation}
\end{proposition}

\begin{proof}
Let $\mathcal B(g)$ denote the right-hand side of \eqref{eq:boundary}.
Fix $B\subseteq\F_q$ with $|B|=k-2$. A hyperplane through
$P_\infty$ and the finite arc points indexed by $B$ has $z=0$,
and its polynomial $g$ is divisible by $P_B$. Thus the corresponding
pencil is
\[
(g,z)=\bigl(P_B(X)(\alpha+\beta X),0\bigr),
 \qquad [\alpha:\beta]\in\PG(1,q).
\]
For $\beta=0$, the roots of $g$ are exactly $B$. If $\beta\ne0$,
the only possible additional root is $-\alpha/\beta$, which already
belongs to $B$ precisely when $\alpha+\beta b=0$ for some $b\in B$.
Consequently, the tangent members of the pencil are
$\beta=0$ and $\alpha+\beta b=0$ ($b\in B$). By Theorem~\ref{thm:envelope},
\begin{equation}\label{eq:phi-on-root-line}
 \Phi\bigl(P_B(X)(\alpha+\beta X),0\bigr)
 =\lambda_B\left(\beta\prod_{b\in B}(\alpha+\beta b)\right)^\epsilon,
 \qquad \lambda_B\in\F_q^\times.
\end{equation}

Applying the discriminant product formula to $P_B$ and
$\alpha+\beta X$, we obtain
\[
\begin{aligned}
\delta_{k-1}\bigl(P_B(X)(\alpha+\beta X)\bigr)
 &=\delta_{k-2}(P_B)\prod_{b\in B}(\alpha+\beta b)
 &&(q\text{ even}),\\
\Disc_{k-1}\bigl(P_B(X)(\alpha+\beta X)\bigr)
 &=\Disc_{k-2}(P_B)\prod_{b\in B}(\alpha+\beta b)^2
 &&(q\text{ odd}).
\end{aligned}
\]
These are polynomial identities in $\alpha,\beta$, so they also
hold when $\beta=0$. Since the coefficient of $X^{k-1}$ in
$P_B(X)(\alpha+\beta X)$ is $\beta$, we obtain
\[
 \mathcal B\bigl(P_B(X)(\alpha+\beta X)\bigr)
 =\begin{cases}
 \delta_{k-2}(P_B)\,\beta\displaystyle\prod_{b\in B}(\alpha+\beta b),
   &q\text{ even},\\[3pt]
 \Disc_{k-2}(P_B)\,\beta^2\displaystyle\prod_{b\in B}(\alpha+\beta b)^2,
   &q\text{ odd}.
 \end{cases}
\]
Both prefactors are nonzero because $P_B$ is square-free.
Comparison with \eqref{eq:phi-on-root-line} shows that
$\Phi(g,0)=c_B\mathcal B(g)$ on this pencil for some $c_B\ne0$.

If $B,B'$ differ in one element, the two pencils contain the same
monic polynomial $P_{B\cup B'}$, for which $\mathcal B$ is nonzero.
Hence $c_B=c_{B'}$. The graph of $(k-2)$-subsets with edges joining
pairs that differ in one element is connected. All $c_B$ are therefore
equal to a single $c\in\F_q^\times$, and for any $T\subseteq\F_q$ with $|T|=k-1$, we have
\[
\Phi(P_T,0)=c\mathcal B(P_T).
\]
Both sides have coefficient degree $\epsilon(k-1)$.
Since $q\ge(\epsilon+1)(k-1)$, Lemma~\ref{lem:coefficient-interpolation}
with $e=\epsilon(k-1)$ and $r=k-1$ shows that
$\Phi(g,0)=c\mathcal B(g)$ for every $g\in\F_q[X]$ of degree
at most $k-1$. Rescaling $\Phi$ by $c^{-1}$ proves the proposition.
\end{proof}

\subsection{The First Transverse Coefficient}

We determine the first transverse coefficient of the envelope
from a missing-value identity for local quotients.

Let $U\in\F_q[X]$, $\deg U<q$, be the canonical representative
of a deep hole of $\RS_q(k)$, where $3\le k\le q-2$.
Fix $S\subseteq\F_q$ with $|S|=k-1$, and let
\[
I_S(X)
=
\sum_{a\in S}
U(a)\frac{P_S(X)}{(X-a)P_S'(a)}
\]
be the polynomial of degree at most $k-2$ interpolating $U$ on $S$. Since $U-I_S$ vanishes on $S$, we may define
\[
Q_S(X)=\frac{U(X)-I_S(X)}{P_S(X)} \in \F_q[X].
\]
The restriction of $Q_S$ to $\F_q\setminus S$ is injective. Indeed, if
$Q_S(x)=Q_S(y)=c$ for distinct $x,y\notin S$, then the polynomial
$I_S+cP_S$, of degree at most $k-1$, agrees with $U$ on the $k+1$ points
of $S\cup\{x,y\}$, contradicting \eqref{eq:deep-hole-criterion}. Hence the set
\[
M_S
=
\F_q\setminus\{Q_S(x):x\in\F_q\setminus S\}
\]
has cardinality $k-1$.

The first-order term of the tangent-pencil product will involve the sum
of the values omitted by $Q_S$. The next lemma computes this sum.

\begin{lemma}\label{lem:missing}
    Let $q$ be a prime power, $3\le k\le q-2$, and let
    $U\in\F_q[X]$, $\deg U<q$, be the canonical representative
    of a deep hole of $\RS_q(k)$. For every $(k-1)$-subset
    $S\subseteq\F_q$, let $M_S$ be defined as above. Then
	\begin{equation}\label{eq:missing}
		\sum_{m\in M_S}m
		=
		\sum_{a\in S}\frac{U'(a)}{P_S'(a)}.
	\end{equation}
\end{lemma}

\begin{proof}
	Since $\deg Q_S\le q-k<q-1$, the standard power-sum identity
	$\sum_{x\in\F_q}H(x)=0$ for $\deg H<q-1$ implies
	\[
	\sum_{x\in\F_q}Q_S(x)=0.
	\]
	Since $Q_S$ is injective on $\F_q\setminus S$ and
    $\sum_{t\in\F_q}t=0$,
	\[
	\sum_{m\in M_S}m
	=
	-\sum_{x\notin S}Q_S(x)
	=
	\sum_{a\in S}Q_S(a).
	\]
	Differentiating $U-I_S=P_SQ_S$ at $a\in S$ yields
	$Q_S(a)
	=
	\frac{U'(a)-I_S'(a)}{P_S'(a)}.$
	Therefore
	\[
	\sum_{m\in M_S}m
	=
	\sum_{a\in S}\frac{U'(a)}{P_S'(a)}
	-
	\sum_{a\in S}\frac{I_S'(a)}{P_S'(a)}.
	\]
	The second sum vanishes by \eqref{eq:lagrange}, since
	$\deg I_S'\le k-3$.
\end{proof}

Under the hypotheses of Proposition~\ref{prop:boundary},
fix the envelope $\Phi$ normalized as in that proposition.
The dual pencil determined by the finite arc points indexed by $S$ is
\[
(g,z)=(\alpha P_S-zI_S,z),
\qquad
[\alpha:z]\in\PG(1,q).
\]
Indeed, a hyperplane $(g,z)$ contains these points precisely when
$g+zI_S$ vanishes on $S$, and hence is a scalar multiple of $P_S$.

For $x\notin S$,
\[
g(x)+zU(x)
=
P_S(x)\bigl(\alpha+zQ_S(x)\bigr),
\]
so the corresponding hyperplane contains the additional finite arc point indexed by $x$ if and only if
\[
\alpha+zQ_S(x)=0.
\]
The member $z=0$ contains the marked point $P_\infty$ and is therefore not tangent to $S$. For $z\ne0$, injectivity of $Q_S$ shows that the tangent members are precisely
\[
[\alpha:z]=[-m:1],
\qquad
m\in M_S.
\]
By Theorem~\ref{thm:envelope},
\[
\Phi(\alpha P_S-zI_S,z)
	=
	\kappa_S
	\prod_{m\in M_S}(\alpha+zm)^\epsilon
\]
for some $\kappa_S\in\F_q^\times$. Evaluating at $(\alpha,z)=(1,0)$ and using Proposition~\ref{prop:boundary}, we obtain
\[
\kappa_S=
\begin{cases}
	\delta_{k-1}(P_S),&q\text{ even},\\
	\Disc_{k-1}(P_S),&q\text{ odd}.
\end{cases}
\]
Therefore
\begin{equation}\label{eq:pencil}
	\Phi(P_S-zI_S,z)
	=
	\begin{cases}
		\delta_{k-1}(P_S)\displaystyle\prod_{m\in M_S}(1+zm),
		& q\text{ even},\\[4pt]
		\Disc_{k-1}(P_S)\displaystyle\prod_{m\in M_S}(1+zm)^2,
		& q\text{ odd}.
	\end{cases}
\end{equation}

The identity \eqref{eq:pencil} describes the variation of $\Phi$ along
the pencil, where both $g$ and $z$ change. To isolate the coefficient of
$z$ with $g$ fixed, we must also compute the first-order variation of
the boundary forms appearing in Proposition~\ref{prop:boundary}.
For a form $F$ in the coefficients of
$g(X)=\sum_{i=0}^{k-1}a_iX^i$ and a direction
$H(X)=\sum_{i=0}^{k-1}h_iX^i$, its coefficient derivative is
\[
 dF(g)[H]=\sum_{i=0}^{k-1}
 \frac{\partial F}{\partial a_i}(g)h_i.
\]
Equivalently, $F(g+\eta H)=F(g)+\eta\,dF(g)[H]$ in
$\F_q[\eta]/(\eta^2)$. We also use the second Hasse derivative,
defined by
\[
 g(X+Z)=g(X)+g'(X)Z+g^{[2]}(X)Z^2+\cdots.
\]

We compute these derivatives at $P_S$ by tracking its roots
under the perturbation $P_S+\eta H$.

\begin{lemma}\label{lem:boundary-differentials}
Let $q$ be a prime power, $3\le k\le q-2$, and let
$S\subseteq\F_q$ have size $k-1$. For any
$H\in\F_q[X]$ with $\deg H<k-1$, in even characteristic,
\begin{equation}\label{eq:delta-differential}
 \frac{d\delta_{k-1}(P_S)[H]}{\delta_{k-1}(P_S)}
 =\sum_{a\in S}\frac{H(a)P_S^{[2]}(a)}{P_S'(a)^2},
\end{equation}
whereas in odd characteristic,
\begin{equation}\label{eq:disc-differential}
 \frac{d\Disc_{k-1}(P_S)[H]}{\Disc_{k-1}(P_S)}
 =-\sum_{a\in S}\frac{H(a)P_S''(a)}{P_S'(a)^2}.
\end{equation}
\end{lemma}

\begin{proof}
Work in $\F_q[\eta]/(\eta^2)$. Since $\deg H<k-1$,
the perturbation $P_S+\eta H$ remains monic.
Its roots are $a+\eta v_a$, where
\begin{equation}\label{eq:root-velocity}
 v_a=-\frac{H(a)}{P_S'(a)}.
\end{equation}
Indeed, substituting $a+\eta v_a$ into $P_S+\eta H$ yields
$\eta(v_aP_S'(a)+H(a))=0$. Moreover, by Lagrange interpolation,
\[
 \prod_{a\in S}(X-a-\eta v_a)
 =P_S(X)-\eta\sum_{a\in S}v_a\frac{P_S(X)}{X-a}
 =P_S(X)+\eta H(X),
\]
so these are all its roots, counted with multiplicity.

Order $S=\{a_1,\ldots,a_{k-1}\}$. Expanding the product of
root differences to first order yields
\[
\begin{aligned}
 \prod_{i<j}\bigl((a_i-a_j)+\eta(v_{a_i}-v_{a_j})\bigr)&=\prod_{i<j}(a_i-a_j)
   \prod_{i<j}\left(1+\eta\frac{v_{a_i}-v_{a_j}}{a_i-a_j}\right)\\
 &=\prod_{i<j}(a_i-a_j)
   \left(1+\eta\sum_{i<j}
   \frac{v_{a_i}-v_{a_j}}{a_i-a_j}\right),
\end{aligned}
\]
where the last equality follows from $\eta^2=0$.
If $q$ is even, the root-difference product is
$\delta_{k-1}(P_S)$ and $a_i-a_j=a_i+a_j$. On the other hand,
by the coefficient-differential definition,
\[
 \delta_{k-1}(P_S+\eta H)
 =\delta_{k-1}(P_S)+\eta\,d\delta_{k-1}(P_S)[H].
\]
Comparing the coefficients of $\eta$ in these two expressions, we find
\[
\frac{d\delta_{k-1}(P_S)[H]}{\delta_{k-1}(P_S)}
 =\sum_{\{a,b\}\subseteq S}\frac{v_a+v_b}{a+b}.
\]
Collecting the terms involving each $v_a$ and using
\eqref{eq:root-velocity} (with $-1=1$), we obtain
\begin{equation}\label{eq:delta-single-variation}
 \frac{d\delta_{k-1}(P_S)[H]}{\delta_{k-1}(P_S)}
 =\sum_{a\in S}v_a\sum_{b\in S\setminus\{a\}}\frac1{a+b}
 =\sum_{a\in S}\frac{H(a)}{P_S'(a)}
   \sum_{b\in S\setminus\{a\}}\frac1{a+b}.
\end{equation}

For any $a\in S$, the expansion $P_S(a+Z)=ZP_S'(a)
 \prod_{b\in S\setminus\{a\}}
 \left(1+\frac{Z}{a-b}\right)$
shows, on comparing the coefficients of $Z^2$, that
\begin{equation}\label{eq:hasse}
 \frac{P_S^{[2]}(a)}{P_S'(a)}
 =\sum_{b\in S\setminus\{a\}}\frac1{a-b}.
\end{equation}
In even characteristic, \eqref{eq:delta-single-variation}
and \eqref{eq:hasse} prove \eqref{eq:delta-differential}.

If $q$ is odd, the discriminant is the square of the
root-difference product. Its relative first-order variation
is therefore twice the sum appearing in the product expansion:
\[
 \frac{d\Disc_{k-1}(P_S)[H]}{\Disc_{k-1}(P_S)}
 =2\sum_{i<j}\frac{v_{a_i}-v_{a_j}}{a_i-a_j}=2\sum_{a\in S}v_a
   \sum_{b\in S\setminus\{a\}}\frac1{a-b}=-\sum_{a\in S}\frac{H(a)P_S''(a)}{P_S'(a)^2},
\]
where the last equality follows from \eqref{eq:root-velocity},
\eqref{eq:hasse}, and $P_S''=2P_S^{[2]}$.
This proves \eqref{eq:disc-differential}.
\end{proof}
Expand the normalized envelope in the transverse coordinate:
\[
\Phi(g,z)
=
\sum_{j=0}^{\epsilon(k-1)}z^j\Phi_j(g),
\]
where $\Phi_j$ is either zero or homogeneous of coefficient degree
$\epsilon(k-1)-j$.

The missing-value identity and the boundary differentials determine
$\Phi_1(P_S)$ through a first-order comparison in \eqref{eq:pencil}.

\begin{proposition}\label{prop:first}
Let $q$ be a prime power, $3\le k\le q-2$, and let
$U\in\F_q[X]$, $\deg U<q$, be the canonical representative
of a deep hole of $\RS_q(k)$. Suppose
$q\ge(\epsilon+1)(k-1)$, and normalize $\Phi$ as in
Proposition~\ref{prop:boundary}. For every $(k-1)$-subset
$S\subseteq\F_q$,
\begin{equation}\label{eq:first-even}
 \Phi_1(P_S)
 =\delta_{k-1}(P_S)
 \left(\sum_{a\in S}\frac{U'(a)}{P_S'(a)}
 +\sum_{a\in S}\frac{U(a)P_S^{[2]}(a)}{P_S'(a)^2}\right)
 \qquad(q\text{ even}),
\end{equation}
and
\begin{equation}\label{eq:first-odd}
 \Phi_1(P_S)
 =\Disc_{k-1}(P_S)
 \left(2\sum_{a\in S}\frac{U'(a)}{P_S'(a)}
 -\sum_{a\in S}\frac{U(a)P_S''(a)}{P_S'(a)^2}\right)
 \qquad(q\text{ odd}).
\end{equation}
The coefficient degree of $\Phi_1$ is at most $k-2$ in even
characteristic and at most $2k-3$ in odd characteristic.
\end{proposition}

\begin{proof}
Fix $S$ and work in $\F_q[\eta]/(\eta^2)$. We substitute
$z=\eta$ in the affine chart $\alpha=1$ of
\eqref{eq:pencil}. Expanding first in the transverse coordinate
and then in the coefficients of $P_S-\eta I_S$, we obtain
\[\Phi(P_S-\eta I_S,\eta)=\Phi_0(P_S-\eta I_S)
   +\eta\Phi_1(P_S-\eta I_S)=\Phi_0(P_S)
   +\eta\bigl(\Phi_1(P_S)-d\Phi_0(P_S)[I_S]\bigr).\]
On the other hand, \eqref{eq:pencil} and
$\Phi_0(P_S)\ne0$ yield
\[\Phi(P_S-\eta I_S,\eta)
 =\Phi_0(P_S)\prod_{m\in M_S}(1+\eta m)^\epsilon=\Phi_0(P_S)
   \left(1+\epsilon\eta\sum_{m\in M_S}m\right).\]
Comparing the coefficients of $\eta$ in
the two identities, we obtain
\begin{equation}\label{eq:first-coefficient-comparison}
 \frac{\Phi_1(P_S)}{\Phi_0(P_S)}
 =\epsilon\sum_{m\in M_S}m
  +\frac{d\Phi_0(P_S)[I_S]}{\Phi_0(P_S)}.
\end{equation}

It remains to evaluate the two terms on the right.
By Lemma~\ref{lem:missing},
\[
 \sum_{m\in M_S}m
 =\sum_{a\in S}\frac{U'(a)}{P_S'(a)}.
\]
By Proposition~\ref{prop:boundary}, $\Phi_0(g)$ is
$a_{k-1}\delta_{k-1}(g)$ in even characteristic and
$a_{k-1}^2\Disc_{k-1}(g)$ in odd characteristic.
Since $P_S$ is monic and $\deg I_S\le k-2$, the leading
coefficient of $P_S-\eta I_S$ remains $1$.
Thus Lemma~\ref{lem:boundary-differentials}, together with
$I_S(a)=U(a)$ for $a\in S$, implies
\[
 \frac{d\Phi_0(P_S)[I_S]}{\Phi_0(P_S)}
 =\begin{cases}
 \displaystyle\sum_{a\in S}
 \frac{U(a)P_S^{[2]}(a)}{P_S'(a)^2},&q\text{ even},\\[9pt]
 \displaystyle-\sum_{a\in S}
 \frac{U(a)P_S''(a)}{P_S'(a)^2},&q\text{ odd}.
 \end{cases}
\]
Substituting into \eqref{eq:first-coefficient-comparison},
with $\epsilon=1$ or $2$ as appropriate, proves
\eqref{eq:first-even} and \eqref{eq:first-odd}.
Finally, $\Phi$ is homogeneous of degree $\epsilon(k-1)$,
so $\Phi_1$ is either zero or homogeneous of coefficient degree
$\epsilon(k-1)-1$, as claimed.
\end{proof}

\section{Deep-Hole Classification via Differential Rigidity}\label{sec:differential}
In this section, we derive differential constraints on $U$
by fixing $k-2$ roots and varying the remaining one in
Proposition~\ref{prop:first}. These constraints yield the
deep-hole classification for $3\le k\le q/2$ when $q$ is even,
and for $3\le k\le(q-1)/2$ with $q>3(k-1)$ when $q$ is odd.

For polynomials $P,U\in\F_q[X]$, write
\[
\mathcal W_P(U)=PU'-P'U.
\]
In characteristic two this is the ordinary derivative of $PU$:
\begin{equation*}
	\mathcal W_P(U)=(PU)'.
\end{equation*}

\subsection{Differential Constraints for Deep Holes}

Let $U$ be the canonical representative of a deep hole
under the hypotheses of Proposition~\ref{prop:first}. Fix a set
$B\subseteq\F_q$ of size $k-2$ and put $P=P_B$. The polynomial $P(Y)(Y-X)$ has degree
$k-1$ in $Y$ and depends affinely on $X$ through its
coefficients. Evaluating $\Phi_1$ at this coefficient vector therefore defines
a polynomial of controlled degree:
\begin{equation}\label{eq:JB-unified}
 J_B(X)=\begin{cases}
 \displaystyle\frac{\Phi_1(P(Y)(Y-X))}{\delta_{k-2}(P)},
     &q\text{ even},\\[8pt]
 \displaystyle\frac{\Phi_1(P(Y)(Y-X))}{\Disc_{k-2}(P)},
     &q\text{ odd}.
 \end{cases}
\end{equation}
The denominators are nonzero, and the coefficient-degree bound for
$\Phi_1$ in Proposition~\ref{prop:first} implies
\[
\deg J_B\le\begin{cases}
 k-2,&q\text{ even},\\
 2k-3,&q\text{ odd}.
 \end{cases}
\]

Apply Proposition~\ref{prop:first} to $S=B\cup\{x\}$ while keeping $B$
fixed. The contribution of the moving root $x$ produces a Wronskian,
whereas the fixed roots contribute a polynomial of bounded degree.

\begin{lemma}\label{lem:moving-root}
Let $q$ be a prime power, $3\le k\le q-2$, and let
$U\in\F_q[X]$, $\deg U<q$, be the canonical representative
of a deep hole of $\RS_q(k)$. Suppose
$q\ge(\epsilon+1)(k-1)$. For every
$B\subseteq\F_q$ with $|B|=k-2$, put $P=P_B$ and define
$J_B$ by \eqref{eq:JB-unified}. Then there is a polynomial
$K_B\in\F_q[X]$ of degree at most $2k-5$ such that,
for every $x\in\F_q\setminus B$,
\begin{align}
 \mathcal W_P(U)(x)&=P(x)J_B(x)+K_B(x)
   &&(q\text{ even}),\label{eq:even-wronskian-fixed}\\
 2\mathcal W_P(U)(x)&=J_B(x)+K_B(x)
   &&(q\text{ odd}).\label{eq:odd-wronskian-fixed}
\end{align}
In even characteristic the first identity extends to all of $\F_q$;
in odd characteristic the second extends to all of $\F_q$ after
multiplication by $P(x)$.
\end{lemma}

\begin{proof}
Fix $B$ and write $P=P_B$. For $b\in B$, put
\[
R_b(X)=\frac{P(X)}{X-b}.
\]
Thus $\deg R_b=k-3$ and $R_b(b)=P'(b)$.
For $j=k-2,k-1$, write $D_j=\delta_j$ in even
characteristic and $D_j=\Disc_j$ in odd characteristic.

Fix $x\in\F_q\setminus B$ and let $S=B\cup\{x\}$, so that
$P_S(Y)=P(Y)(Y-x)$. The root-difference product formulas give
\[
D_{k-1}(P_S)=D_{k-2}(P)P(x)^\epsilon.
\]
The ordinary and Hasse product rules yield
\[
\begin{aligned}
P_S'(x)&=P(x), &
P_S^{[2]}(x)&=P'(x),\\
P_S'(b)&=(b-x)P'(b), &
P_S^{[2]}(b)&=(b-x)P^{[2]}(b)+P'(b)
\quad(b\in B).
\end{aligned}
\]
In both characteristics, Proposition~\ref{prop:first} can be
written in the unified form
\[
\frac{\Phi_1(P_S)}{D_{k-1}(P_S)}
=\epsilon\sum_{a\in S}
\left(\frac{U'(a)}{P_S'(a)}
-\frac{U(a)P_S^{[2]}(a)}{P_S'(a)^2}\right).
\]
Here the even-characteristic formula uses $-1=1$, while
the odd-characteristic formula uses $P_S''=2P_S^{[2]}$.

The summand corresponding to the moving root $x$ is
\[
\epsilon\left(\frac{U'(x)}{P(x)}
-\frac{U(x)P'(x)}{P(x)^2}\right)
=\frac{\epsilon\mathcal W_P(U)(x)}{P(x)^2}.
\]
For each fixed root $b\in B$, define
\[
c_b=\epsilon\left(
\frac{U'(b)}{P'(b)}
-\frac{U(b)P^{[2]}(b)}{P'(b)^2}\right).
\]
Its contribution is
\[
\begin{aligned}
&\epsilon\left(
\frac{U'(b)}{(b-x)P'(b)}
-\frac{U(b)((b-x)P^{[2]}(b)+P'(b))}
{(b-x)^2P'(b)^2}\right)\\
&\qquad=\frac{c_b}{b-x}
-\frac{\epsilon U(b)}{P'(b)(b-x)^2}.
\end{aligned}
\]
By the definition of $J_B$ in \eqref{eq:JB-unified},
we therefore obtain, for $x\notin B$,
\begin{equation}\label{eq:moving-root-unified-step}
\frac{J_B(x)}{P(x)^\epsilon}
=\frac{\epsilon\mathcal W_P(U)(x)}{P(x)^2}
+\sum_{b\in B}\left(
\frac{c_b}{b-x}
-\frac{\epsilon U(b)}{P'(b)(b-x)^2}\right).
\end{equation}
Since $P(x)/(b-x)=-R_b(x)$, multiplying by $P(x)^2$
and rearranging gives
\[
\epsilon\mathcal W_P(U)(x)
=P(x)^{2-\epsilon}J_B(x)+K_B(x),
\]
where
\[
K_B(X)=P(X)\sum_{b\in B}c_bR_b(X)
+\epsilon\sum_{b\in B}
\frac{U(b)}{P'(b)}R_b(X)^2.
\]
As $\deg(PR_b)=2k-5$ and $\deg R_b^2=2k-6$,
we have $\deg K_B\le2k-5$.

If $q$ is even, then $\epsilon=1$ and
$\mathcal W_P(U)=(PU)'$. For $b\in B$,
$P(b)=0$, $R_a(b)=0$ for $a\ne b$, and $R_b(b)=P'(b)$.
Consequently,
\[
K_B(b)=U(b)P'(b)
=-P'(b)U(b)=\mathcal W_P(U)(b),
\]
where the middle equality uses characteristic two.
Thus \eqref{eq:even-wronskian-fixed} also holds at every
$b\in B$.
If $q$ is odd, then $\epsilon=2$; multiplying the
identity by $P(x)$ makes both sides vanish at $x\in B$.
Hence the multiplied form of \eqref{eq:odd-wronskian-fixed}
holds throughout $\F_q$.
\end{proof}

For $0\le m<q$, let
\[
 \E_m=\{\ev(H):H\in\F_q[X],\ \deg H\le m\}.
\]

\begin{corollary}\label{cor:root-multiplier}
Let $q$ be a prime power, $3\le k\le q-2$, and let
$U\in\F_q[X]$, $\deg U<q$, be the canonical representative
of a deep hole of $\RS_q(k)$. Assume
$q\ge(\epsilon+1)(k-1)$ and fix any
$B\subseteq\F_q$ with $|B|=k-2$.
In even characteristic,
\begin{equation}\label{eq:even-roots}
 \ev\bigl(\mathcal W_{P_B}(U)\bigr)\in\E_{2k-4},
\end{equation}
whereas in odd characteristic,
\begin{equation}\label{eq:odd-roots}
 \ev\bigl(P_B\mathcal W_{P_B}(U)\bigr)\in\E_{3k-5}.
\end{equation}
In characteristic two, \eqref{eq:even-roots} is equivalently
written as $\ev((P_BU)')\in\E_{2k-4}$.
\end{corollary}

\begin{proof}
The hypotheses ensure $2k-4<q$ in even characteristic and
$3k-5<q$ in odd characteristic. In even characteristic, the
first identity in Lemma~\ref{lem:moving-root} holds at every
element of $\F_q$, and
\[
 \deg(P_BJ_B)\le2k-4,\qquad \deg K_B\le2k-5.
\]
This proves \eqref{eq:even-roots}. In odd characteristic, multiply
the second identity by $P_B$. The resulting equality holds on the
whole field, and
\[
 \deg\bigl(P_B(J_B+K_B)\bigr)\le3k-5.
\]
Since $2$ is invertible, \eqref{eq:odd-roots} follows.
These are equalities of polynomial functions; no equality of
formal polynomials is asserted.
\end{proof}

The following spanning lemma extends these constraints from root polynomials to all polynomials of the required degrees.

\begin{lemma}\label{lem:root-powers}
Let $q$ be a prime power, $p=\operatorname{char}\F_q$, and
let $1\le d<p$ and $r\ge0$ be integers.
If $q\ge r+d$, then
\[
 \operatorname{span}_{\F_q}\{P_B^d:B\subseteq\F_q,\ |B|=r\}
 =\F_q[X]_{\le dr}.
\]
\end{lemma}

\begin{proof}
Suppose a linear functional
$\ell:\F_q[X]_{\le dr}\to\F_q$ vanishes on $P_B^d$
for every $r$-subset $B$. Then
\[
F(a_0,\ldots,a_r)
=\ell\left(\left(\sum_{i=0}^r a_iX^i\right)^d\right)
\]
is a homogeneous form of degree $d$ vanishing at all $P_B$.
Lemma~\ref{lem:coefficient-interpolation}, with $e=d$, gives
$F\equiv0$. For each $0\le t\le dr$, choose nonnegative integers
$m_0,\ldots,m_r$ such that
$\sum_i m_i=d$ and $\sum_i im_i=t$.
The coefficient of $\prod_i a_i^{m_i}$ in $F$ is
\[
\binom{d}{m_0,\ldots,m_r}\ell(X^t).
\]
Since $d<p$, this multinomial coefficient is nonzero in $\F_q$.
Thus $\ell(X^t)=0$ for every $t$, hence $\ell=0$.
The claimed spanning property follows.
\end{proof}

\begin{proposition}\label{prop:unified-differential}
Let $q$ be a prime power, $3\le k\le q-2$, and let
$U\in\F_q[X]$, $\deg U<q$, be the canonical representative
of a deep hole of $\RS_q(k)$. Assume
$q\ge(\epsilon+1)(k-1)$.
In even characteristic, every $H\in\F_q[X]$ with
$\deg H\le k-2$ satisfies
\begin{equation}\label{eq:even-operator}
 \ev((HU)')\in\E_{2k-4}.
\end{equation}
In odd characteristic, every $H\in\F_q[X]$ with
$\deg H\le2k-4$ satisfies
\begin{equation}\label{eq:odd-operator}
 \ev(2HU'-H'U)\in\E_{3k-5}.
\end{equation}
\end{proposition}

\begin{proof}
When $q$ is even, the identity $\mathcal W_P(U)=(PU)'$ and
Corollary~\ref{cor:root-multiplier} imply
\eqref{eq:even-operator} for $H=P_B$.
By Lemma~\ref{lem:root-powers} with $d=1$ these polynomials span
$\F_q[X]_{\le k-2}$, and the map $H\mapsto(HU)'$ is linear.
This proves the first assertion.

When $q$ is odd, set $H=P_B^2$. Then
\[
2P_B\mathcal W_{P_B}(U)=2HU'-H'U.
\]
The second part of Corollary~\ref{cor:root-multiplier} therefore
proves \eqref{eq:odd-operator} for these squares.
Since $2<\operatorname{char}\F_q$, Lemma~\ref{lem:root-powers}
with $d=2$ shows that they span $\F_q[X]_{\le2k-4}$.
Linearity in $H$ completes the proof.
\end{proof}

The constraints just obtained concern evaluation vectors, not
formal polynomial degrees. We compare leading terms below only
after verifying that the relevant polynomials have degree less than $q$.

\subsection{Degree Bounds and Classification}

For a canonical polynomial $U$, write $u_j=[X^j]U$ for
$0\le j<q$. We first treat characteristic two.

\begin{lemma}\label{lem:even-kernel}
Let $q$ be an even prime power, $k\ge3$, and $q\ge2k$.
Among polynomials $U\in\F_q[X]$ with $\deg U<q$, the
solutions of \eqref{eq:even-operator} form exactly
$\F_q[X]_{\le k}$.
\end{lemma}

\begin{proof}
Every $U$ of degree at most $k$ satisfies the constraint:
for $\deg H\le k-2$ we have $\deg(HU)\le2k-2$,
and the derivative of its leading term, if its degree is
$2k-2$, vanishes in characteristic two.
Thus $\deg(HU)'\le2k-4$.

Conversely, let $U$ be a nonzero canonical solution and write
$D=\deg U<q$. For $H=1$ and $H=X$, the polynomials
$U'$ and $(XU)'$ have degree less than $q$. Their evaluation
constraints therefore imply formal degree bounds.
In characteristic two, $U'$ detects the odd-degree terms of $U$,
while $(XU)'$ detects its even-degree terms; together they imply $D\le2k-3.$

If $k=3$, the preliminary bound already gives $D\le k$.
Suppose henceforth that $k\ge4$. If $D\ge k+1$,
put $h=2k-1-D$; then $2\le h\le k-2$.
The leading term of $(X^hU)'$ has degree
$2k-2<q$ and coefficient
\[
 (D+h)u_D=(2k-1)u_D=u_D\ne0,
\]
contrary to \eqref{eq:even-operator}. Hence $D\le k$.
\end{proof}

The even-characteristic classification follows at once.

\begin{theorem}\label{thm:even}
Let $q$ be an even prime power and $3\le k\le q/2$.
Then a vector is a deep hole of $\RS_q(k)$ if and only if
its canonical representative has degree $k$.
\end{theorem}

\begin{proof}
Since $q\ge2k$, the size hypothesis of
Proposition~\ref{prop:unified-differential} holds. That proposition
and Lemma~\ref{lem:even-kernel} imply $\deg U\le k$ for a deep hole.
A polynomial of degree less than $k$ represents a codeword,
so $\deg U=k$. The converse is the standard degree-$k$ construction.
\end{proof}

For odd characteristic, we now assume $q>3(k-1)$.

\begin{lemma}\label{lem:odd-kernel}
Let $q$ be an odd prime power, $k\ge3$, and $q>3(k-1)$.
Among polynomials $U\in\F_q[X]$ with $\deg U<q$, the
solutions of \eqref{eq:odd-operator} form
\[
 \begin{cases}
 \F_q[X]_{\le k},&\operatorname{char}\F_q\ne3,\\
 \F_q[X]_{\le k+1},&\operatorname{char}\F_q=3.
 \end{cases}
\]
\end{lemma}

\begin{proof}
The indicated spaces satisfy \eqref{eq:odd-operator}.
Indeed, for $\deg U\le k$ the expression
$2HU'-H'U$ has degree at most $3k-5$.
If $\deg U=k+1$ and $\operatorname{char}\F_q=3$,
the only possible term of degree $3k-4$ occurs when
$\deg H=2k-4$. Its coefficient is
$6h_{2k-4}u_{k+1}=0$, where
$h_{2k-4}=[X^{2k-4}]H$.

Conversely, take a nonzero canonical solution $U$ and let
$D=\deg U$. Set $M=3k-5$.
For $H=1$ and $H=X$, the respective candidate leading terms of
$2HU'-H'U$ have degrees $D-1,D$ and coefficients
$2Du_D,(2D-1)u_D$. Both polynomials have degree less than $q$,
and these coefficients cannot both vanish. Their evaluation
constraints force $D\le M+1.$ 

If $D\ge k+2$, put $h=M+2-D$; then
$1\le h\le2k-5$. For $H=X^h$ and $H=X^{h+1}$, the candidate
leading degrees are $M+1,M+2$, with respective coefficients
$(2D-h)u_D$ and $(2D-h-1)u_D$. At least one is nonzero. Since $q>3(k-1)$,
both degrees are less than $q$,
contradicting \eqref{eq:odd-operator}. Thus $D\le k+1$.

If $D=k+1$, take $H=X^{2k-4}$.
Its degree-$3k-4$ coefficient is $6u_D$, which is nonzero
unless $\operatorname{char}\F_q=3$. This proves the assertion.
\end{proof}

\begin{theorem}\label{thm:odd-low}
Let $q$ be an odd prime power, $3\le k\le(q-1)/2$,
and $q>3(k-1)$.
Then a vector is a deep hole of $\RS_q(k)$ if and only if its
canonical representative has degree $k$.
\end{theorem}

\begin{proof}
Let $U$ represent a deep hole.
By Proposition~\ref{prop:unified-differential} and
Lemma~\ref{lem:odd-kernel}, we have $\deg U\le k$,
unless $\operatorname{char}\F_q=3$ and $\deg U=k+1$.
The latter possibility is excluded by
Li and Zhu~\cite[Theorem~1]{LiZhu}, who proved that
every word of degree $k+1$ has distance $q-k-1$
from $\RS_q(k)$ in odd characteristic.
Hence $\deg U=k$, since polynomials of degree
less than $k$ represent codewords.
The converse follows from the degree-$k$ construction.
\end{proof}

\section{Local Syndromes and Template Uniqueness}\label{sec:templates}
In this section, we use global and local syndromes to characterize
deep holes, establish the existence and uniqueness of local templates,
and derive compatibility identities. These results will be used in
Section~\ref{sec:quadratic}.

Throughout this section, $q$ is odd, $3\le k\le(q-1)/2$,
and $U$ is the canonical representative of a deep hole
of $\RS_q(k)$; no restriction $q\le3(k-1)$ is imposed.

\subsection{A Syndrome Characterization of Deep Holes}

\begin{definition}[Global and local syndromes]\label{def:syndromes}
Put $\nu=q-k-1$. The \emph{global syndrome functional} of $U$ is
\[
\Lambda_U:\F_q[X]_{\le\nu}\longrightarrow\F_q,
\qquad
\Lambda_U(F)=\sum_{a\in\F_q}U(a)F(a).
\]
For $1\le m\le\nu$ and $W\subseteq\F_q$ with
$|W|=k+m+1$, the \emph{local syndrome functional} on $W$ is
\[
L_W:\F_q[X]_{\le m}\longrightarrow\F_q,
\qquad L_W(g)=\Lambda_U(P_{\F_q\setminus W}g).
\]
This is well defined since $|\F_q\setminus W|+m=\nu$.
\end{definition}
If $f\in\F_q[X]$ has degree less than $k$ and $\deg F\le\nu$, then
\[
\deg(fF)\le q-2,
\]
and hence $\sum_{a\in\F_q}f(a)F(a)=0$
by the standard power-sum identities over $\F_q$. Thus $\Lambda_U$ depends only on the coset of $\ev(U)$ modulo $\RS_q(k)$.

The deep-hole condition admits a useful formulation in terms of $\Lambda_U$. Let $A\subseteq\F_q$ have size $\nu$, and put
\[
V=\F_q\setminus A,
\qquad |V|=k+1.
\]
Since
\[
X^q-X=P_A(X)P_V(X),
\]
differentiating at $v\in V$, we obtain
\[
P_A(v)P_V'(v)=-1.
\]
Therefore
\begin{equation}\label{eq:global-zero-free}
	\Lambda_U(P_A)
	=
	-\sum_{v\in V}\frac{U(v)}{P_V'(v)}
	=
	-[X^k]I_V,
\end{equation}
where $I_V$ is the polynomial of degree at most $k$ interpolating $U$ on $V$. By \eqref{eq:deep-hole-criterion},
\begin{equation}\label{eq:global-ZF-criterion}
	\ev(U)\text{ is a deep hole}
	\quad\Longleftrightarrow\quad
	\Lambda_U(P_A)\ne0
	\quad
	\text{for every }A\in\binom{\F_q}{\nu}.
\end{equation}

Fix $1\le m\le\nu$ and $W\subseteq\F_q$ with
$|W|=k+m+1$. Differentiating
$X^q-X=P_{\F_q\setminus W}P_W$ at $x\in W$, we find
\[
P_{\F_q\setminus W}(x)=-\frac{1}{P_W'(x)}.
\]
Hence
\begin{equation}\label{eq:global-localization}
	L_W(g)
	=
	-\sum_{x\in W}
	\frac{U(x)g(x)}{P_W'(x)}.
\end{equation}

More generally, for $\mathbf y=(y_x)_{x\in W}\in\F_q^W$, define
\[
\Psi_W(\mathbf y)(g)
	:=
	-\sum_{x\in W}
	\frac{y_xg(x)}{P_W'(x)},
	\qquad
	g\in\F_q[X]_{\le m}.
\]
Thus
\[
\Psi_W:\F_q^W\longrightarrow(\F_q[X]_{\le m})^*
\]
is linear, and
\[
L_W=\Psi_W\bigl((U(x))_{x\in W}\bigr).
\]

\begin{lemma}\label{lem:local}
Let $U$ represent a deep hole of $\RS_q(k)$, let $1\le m\le q-k-1$, and let $W\subseteq\F_q$ have size $k+m+1$. Then $\Psi_W$ has rank $m+1$ and
\[
\ker\Psi_W=\{(f(x))_{x\in W}:\deg f<k\}.
\]
Moreover, for every $m$-subset $T\subseteq W$, writing $V=W\setminus T$, we have
\begin{equation}\label{eq:delete-roots}
\begin{aligned}
L_W(P_T)&=\Lambda_U(P_{(\F_q\setminus W)\cup T})
=-\sum_{v\in V}\frac{U(v)}{P_V'(v)}\\
&=-[X^k]I_V\ne0.
\end{aligned}
\end{equation}
In particular, $L_W\ne0$.
\end{lemma}

\begin{proof}
	Relative to the monomial basis $1,X,\ldots,X^m$ of $\F_q[X]_{\le m}$, the matrix of $\Psi_W$ is
	\[
	\left(
	-\frac{x^j}{P_W'(x)}
	\right)_{\substack{0\le j\le m\\ x\in W}}.
	\]
	This is a truncated Vandermonde matrix multiplied by an invertible diagonal matrix. Hence
	\[
	\rank\Psi_W=m+1.
	\]

	If $\mathbf y=(f(x))_{x\in W}$ for some $\deg f<k$ and $\deg g\le m$, then
	\[
	\deg(fg)\le(k-1)+m=|W|-2.
	\]
	By \eqref{eq:lagrange},
	\[
	\Psi_W(\mathbf y)(g)
	=
	-\sum_{x\in W}
	\frac{f(x)g(x)}{P_W'(x)}
	=
	0.
	\]
	Thus the $k$-dimensional space of degree-$<k$ evaluation vectors on $W$ lies in $\ker\Psi_W$. Since
\[
\dim\ker\Psi_W=|W|-(m+1)=k,
\]
these two spaces coincide.

	Finally, since $|\F_q\setminus W|=\nu-m$ and $T\subseteq W$, the disjoint union $(\F_q\setminus W)\cup T$ has size $\nu$. Consequently
\[
L_W(P_T)=\Lambda_U(P_{(\F_q\setminus W)\cup T})\ne0
\]
by \eqref{eq:global-ZF-criterion}. Its complementary set is $V=W\setminus T$, so the remaining identities in \eqref{eq:delete-roots} follow from \eqref{eq:global-zero-free}.
\end{proof}

\subsection{Template Characterization and Compatibility Identities}

Two types of templates describe the evaluation-type local syndromes.

\begin{definition}\label{def:template}
	A \emph{template} is either
	\begin{enumerate}
		\item a polynomial $G\in\F_q[X]$ of degree exactly $k$, or
		\item a rational function
		\[
		R(X)=f(X)+\frac{c}{X-\zeta},
		\qquad
		\deg f<k,\quad c\in\F_q^\times,\quad \zeta\in\F_q,
		\]
		defined on $\F_q\setminus\{\zeta\}$.
	\end{enumerate}
	For a template $R$, write $\dom R=\F_q$ in the polynomial
	case and $\dom R=\F_q\setminus\{\zeta\}$ in the case of a
	template with pole $\zeta$.
	Relative to $U$, a subset $W\subseteq\F_q$ is \emph{good} if there exists a template $R$, defined on $W$, such that
	\[
	U(x)=R(x)
	\qquad
	(x\in W).
	\]
	Otherwise $W$ is \emph{bad}.
\end{definition}

Let $1\le m\le\nu$, and let $W\subseteq\F_q$ have size $k+m+1$. For $\zeta\in\F_q$, define the evaluation functional
\[
\lambda_\zeta(g)=g(\zeta),
\qquad
g\in\F_q[X]_{\le m},
\]
and let
\[
\lambda_\infty(g)=[X^m]g.
\]
With respect to the monomial basis $1,X,\ldots,X^m$, their projective classes form
\[
\C_m
=
\{[1:\zeta:\cdots:\zeta^m]:\zeta\in\F_q\}
\cup
\{[0:\cdots:0:1]\}
\subseteq\PG(m,q).
\]

\begin{lemma}\label{lem:template-criterion}
Let $U$ represent a deep hole of $\RS_q(k)$, $1\le m\le q-k-1$, and $|W|=k+m+1$. Then $W$ is good if and only if
	\[
	[L_W]\in\C_m.
	\]
\end{lemma}

\begin{proof}
The local syndrome is nonzero by Lemma~\ref{lem:local}. By Lagrange interpolation,
\begin{align}
 \Psi_W\bigl((x^k)_{x\in W}\bigr)&=-\lambda_\infty,
 \notag\\
 \Psi_W\left(\left(\frac1{x-\zeta}\right)_{x\in W}\right)
 &=\frac{\lambda_\zeta}{P_W(\zeta)},\qquad \zeta\notin W.
 \label{eq:pole-evaluation}
\end{align}
For the first identity, apply \eqref{eq:lagrange} to
$H=X^kg$ with $\deg g\le m$ and $|W|=k+m+1$:
\[
\begin{aligned}
\Psi_W\bigl((x^k)_{x\in W}\bigr)(g)
&=-\sum_{x\in W}\frac{x^kg(x)}{P_W'(x)}\\
&=-[X^{k+m}](X^kg)
=-[X^m]g=-\lambda_\infty(g).
\end{aligned}
\]
For the second identity, evaluating the Lagrange formula at
$\zeta\notin W$ yields
\[
 \frac{g(\zeta)}{P_W(\zeta)}
 =\sum_{x\in W}\frac{g(x)}{(\zeta-x)P_W'(x)},
\]
which is precisely \eqref{eq:pole-evaluation}.

By Lemma~\ref{lem:local}, the kernel of $\Psi_W$ consists of evaluations of polynomials of degree less than $k$. Thus the first identity says that $[L_W]=[\lambda_\infty]$ exactly when $U|_W$ is a nonzero multiple of $X^k$ plus a polynomial of degree less than $k$. The second says that $[L_W]=[\lambda_\zeta]$, with $\zeta\notin W$, exactly when $U|_W$ is a nonzero multiple of $(X-\zeta)^{-1}$ plus such a polynomial. These are precisely the two template types.

Finally, $[L_W]=[\lambda_\zeta]$ with $\zeta\in W$ is impossible: an $m$-subset $T\subseteq W$ containing $\zeta$ would satisfy $L_W(P_T)=0$, contrary to \eqref{eq:delete-roots}.
\end{proof}

The following intersection bound ensures uniqueness of local templates that agree on at least $k+2$ points.

\begin{lemma}\label{lem:template-intersection}
Two distinct templates agree on at most $k+1$ points of their common domain.
\end{lemma}
\begin{proof}
Write each template in reduced form $N_i/D_i$, with $\deg N_i\le k$ and $\deg D_i\le1$. For distinct templates, $N_1D_2-N_2D_1$ is a nonzero polynomial of degree at most $k+1$. Every point of agreement in the common domain is a root of this polynomial.
\end{proof}

Although this section concerns odd characteristic and $k\ge3$,
the following local algebraic assertion holds for every prime
power $q$ and every $1\le k\le q-2$.

\begin{corollary}\label{cor:unique-template}
Let $U$ be the canonical representative of a deep hole of $\RS_q(k)$, with $1\le k\le q-2$. Every $(k+2)$-subset of $\F_q$ has a unique template agreeing with $U$ on that subset.
\end{corollary}
\begin{proof}
The proofs of Lemmas~\ref{lem:local},
\ref{lem:template-criterion}, and~\ref{lem:template-intersection}
apply without change to every prime power $q$ and
$1\le k\le q-2$.
For $m=1$, the evaluation curve $\C_1$ is the entire
projective line, so existence follows from
Lemma~\ref{lem:template-criterion}.
Uniqueness follows from Lemma~\ref{lem:template-intersection}.
\end{proof}

Local syndromes on overlapping sets satisfy a deletion identity
used below to count failed template extensions.

\begin{lemma}\label{lem:compatibility}
Let $2\le m\le q-k-1$, $|W|=k+m+1$, and $z\in W$. Then
\begin{equation}\label{eq:compatibility}
 L_{W\setminus\{z\}}(g)=L_W((X-z)g)
 \qquad(\deg g\le m-1).
\end{equation}
\end{lemma}
\begin{proof}
Since
\[
P_{\F_q\setminus(W\setminus\{z\})}
=(X-z)P_{\F_q\setminus W},
\]
the definition of the local syndrome gives
\[
L_{W\setminus\{z\}}(g)
=\Lambda_U\bigl((X-z)P_{\F_q\setminus W}g\bigr)
=L_W((X-z)g).
\]
\end{proof}

\section{Deep-Hole Classification via Quadratic Obstructions}\label{sec:quadratic}
In this section, we combine local-syndrome compatibility with
exact counts of quadratic obstructions to exclude nonexpected
deep holes in the remaining odd-characteristic range for $q\ge41$.
The small-field cases are treated in Section~\ref{sec:completion}.
Unless stated otherwise, throughout this section $q$
is an odd prime power, $3\le k\le(q-1)/2$, and $U$
is the canonical representative of a deep hole of $\RS_q(k)$.

\subsection{Counting Quadratic Obstructions}

A bad $(k+3)$-subset $W$ has $[L_W]\notin\C_2$ by
Lemma~\ref{lem:template-criterion}. We associate with such
a syndrome a graph on the evaluation points.

\begin{definition}[Quadratic obstruction graph]\label{def:quadratic-graph}
For $0\ne L\in(\F_q[X]_{\le2})^*$, let $\mathcal H_L$
be the simple graph on $\F_q$ whose edges are the unordered
pairs $\{x,y\}$ of distinct elements satisfying
$L((X-x)(X-y))=0$. This graph depends only on the
projective class $[L]=[a_0:a_1:a_2]$, where
$a_i=L(X^i)$.
\end{definition}

Expanding $(X-x)(X-y)$ gives the edge equation
\begin{equation}\label{eq:quadratic-edge}
	a_0xy-a_1(x+y)+a_2=0.
\end{equation}
For bad $W$, identity~\eqref{eq:delete-roots} shows that
$W$ is independent in $\mathcal H_{L_W}$.

Suppose $[L]\notin\C_2$, so that
$a_0a_2-a_1^2\ne0$. If $a_0\ne0$, normalize $a_0=1$
and put $\alpha=a_1$ and $\Delta=\alpha^2-a_2\ne0$.
Then \eqref{eq:quadratic-edge} becomes
\[
(x-\alpha)(y-\alpha)=\Delta.
\]
The vertex $\alpha$ is isolated. For each $x\ne\alpha$,
the unique possible neighbor is
\[
\sigma(x)=\alpha+\frac{\Delta}{x-\alpha}.
\]
Since $\sigma^2=\mathrm{id}$, its two-element orbits
form disjoint edges, while its fixed points are isolated,
as loops are excluded. The fixed-point equation
\[
(x-\alpha)^2=\Delta
\]
has two solutions if $\Delta$ is a square in
$\F_q^\times$, and none otherwise. Thus $\mathcal H_L$
has three isolated vertices and $(q-3)/2$ disjoint edges
in the former case, and one isolated vertex and
$(q-1)/2$ disjoint edges in the latter.

If $a_0=0$, then $a_1\ne0$, and
\eqref{eq:quadratic-edge} reduces to
\[
x+y=c,\qquad c=\frac{a_2}{a_1}.
\]
The involution $x\mapsto c-x$ has the unique fixed
point $c/2$. Hence $\mathcal H_L$ has one isolated
vertex and $(q-1)/2$ disjoint edges.

For $i\in\{1,3\}$, let $N_i$ denote the number of
non-evaluation projective classes $[L]$ for which
$\mathcal H_L$ has exactly $i$ isolated vertices.
In the chart $a_0=1$, there are $q$ choices of
$\alpha$ and $(q-1)/2$ choices of $\Delta$ of each
quadratic character. The chart $a_0=0$ contributes
another $q$ classes, all with one isolated vertex.
Consequently,
\[
N_1=\frac{q(q+1)}2,\qquad
N_3=\frac{q(q-1)}2,
\]
and
\[
N_1+N_3=q^2=|\PG(2,q)\setminus\C_2|.
\]

We can now count independent subsets in these matchings
exactly. For $i\in\{1,3\}$, a graph with $i$ isolated vertices
has $(q-i)/2$ disjoint edges. An independent $w$-subset
containing $j$ isolated vertices is obtained by choosing
$j$ isolated vertices, $w-j$ edges, and one endpoint
from each chosen edge. Hence the number of independent
$w$-subsets is
\begin{equation}\label{eq:Ai}
	A_i(q,w)=
	\sum_{j=0}^{\min\{i,w\}}
	\binom ij 2^{w-j}\binom{(q-i)/2}{w-j},
\end{equation}
where out-of-range binomial coefficients are zero.

Define
\begin{equation}\label{eq:Q}
	Q(q,w)=
	\frac{N_1A_1(q,w)+N_3A_3(q,w)}{\binom qw},
	\qquad 0\le w\le q.
\end{equation}
The formulas \eqref{eq:Ai} and \eqref{eq:Q}, with
$N_1=q(q+1)/2$ and $N_3=q(q-1)/2$, also define
$Q(q,w)$ for every odd integer $q\ge3$.
We use this numerical extension in
Appendix~\ref{app:quadratic}; the finite-geometric
interpretation applies when $q$ is a prime power.
For a uniformly random $w$-subset $W\subseteq\F_q$,
\[
\Pr(W\text{ is independent in }\mathcal H_L)
=\frac{A_i(q,w)}{\binom qw}
\]
whenever $\mathcal H_L$ has $i$ isolated vertices.
By linearity of expectation,
\[
Q(q,w)=\mathbb E\left[
\#\{[L]\notin\C_2:
W\text{ is independent in }\mathcal H_L\}
\right].
\]
Thus $Q(q,w)$ counts projective classes in expectation
and may exceed one.

\begin{lemma}\label{lem:Q-monotone}
	For each odd prime power $q$, the function $Q(q,w)$ is
	nonincreasing in $w$.
\end{lemma}

\begin{proof}
	Let $W_{w+1}$ be a uniformly random $(w+1)$-subset,
	and obtain $W_w$ by deleting a uniformly chosen element.
	Then $W_w$ is uniformly distributed among all
	$w$-subsets. Since independence is preserved under
	deletion, for every $[L]\notin\C_2$,
	\[
	\Pr(W_{w+1}\text{ is independent in }\mathcal H_L)
	\le
	\Pr(W_w\text{ is independent in }\mathcal H_L).
	\]
	Summing over these projective classes gives
	$Q(q,w+1)\le Q(q,w)$.
\end{proof}

\subsection{Excluding Nonexpected Deep Holes}

The deletion identity in Section~\ref{sec:templates} first shows
that distinct failed extensions give distinct quadratic obstructions.

\begin{lemma}\label{lem:noncollision}
Let $q$ be odd, $3\le k\le(q-1)/2$, and let $U$ represent a deep hole of $\RS_q(k)$. Fix a $(k+2)$-subset $B$. If $z,z'\notin B$ are distinct and both $B\cup\{z\}$ and $B\cup\{z'\}$ are bad, then
\[
 [L_{B\cup\{z\}}]\ne[L_{B\cup\{z'\}}].
\]
\end{lemma}
\begin{proof}
Write $F_z=L_{B\cup\{z\}}$ and suppose $F_z=cF_{z'}$ for some $c\in\F_q^\times$. For every $\deg g\le1$, applying \eqref{eq:compatibility} to both one-point deletions yields
\[
 F_z((X-z)g)=L_B(g)=F_{z'}((X-z')g).
\]
Consequently
\begin{equation}\label{eq:noncollision-annihilator}
 F_{z'}\bigl(((c-1)X+z'-cz)g\bigr)=0
 \qquad(\deg g\le1).
\end{equation}
The linear polynomial in parentheses is nonzero, since $z\ne z'$. Writing it as $\alpha X+\beta$, with $(\alpha,\beta)\ne(0,0)$, and taking $g=1,X$ in \eqref{eq:noncollision-annihilator}, we obtain
\[
\begin{pmatrix}
 F_{z'}(1)&F_{z'}(X)\\
 F_{z'}(X)&F_{z'}(X^2)
\end{pmatrix}
\begin{pmatrix}\beta\\\alpha\end{pmatrix}=0.
\]
Thus this Hankel matrix is singular.
Writing the moments of $F_{z'}$ as $(a_0,a_1,a_2)$, this equality says $a_0a_2=a_1^2$. Since $F_{z'}\ne0$, it follows that $[F_{z'}]$ lies on $\C_2$: for $a_0\ne0$ it is evaluation at $a_1/a_0$, and for $a_0=0$ it is evaluation at infinity. This contradicts the badness of $B\cup\{z'\}$ by Lemma~\ref{lem:template-criterion}.
\end{proof}

For each $(k+2)$-subset $B\subseteq\F_q$, let $R_B$
be its unique template, given by
Corollary~\ref{cor:unique-template}, and put
\[
T_B=\{x\in\dom R_B:U(x)=R_B(x)\}.
\]
Thus $B\subseteq T_B$. To count the quadratic
obstructions associated with $B$, define
\[
N(B)=\#\{[L]\in\PG(2,q)\setminus\C_2:
B\text{ is independent in }\mathcal H_L\}.
\]

\begin{lemma}\label{lem:failure-multiplicity}
	Let $q$ be odd, $3\le k\le(q-1)/2$, and let $U$
	represent a deep hole of $\RS_q(k)$. Then, for every
	$(k+2)$-subset $B$,
	\begin{equation}\label{eq:failure-multiplicity}
		N(B)\ge q-|T_B|.
	\end{equation}
\end{lemma}

\begin{proof}
	For each $z\in\F_q\setminus T_B$, the extension
	$W_z=B\cup\{z\}$ is bad. Indeed, if $W_z$ admitted
	a template, it would agree with $R_B$ on the $k+2$
	points of $B$ and hence coincide with $R_B$ by
	Lemma~\ref{lem:template-intersection}. This would
	force $z\in T_B$, a contradiction. The argument
	also covers the pole of a rational template, where
	$R_B$ is undefined.
	
	Since $W_z$ is bad,
	$[L_{W_z}]\notin\C_2$. Moreover, by
	\eqref{eq:delete-roots}, for distinct $x,y\in B$,
	\[
	L_{W_z}((X-x)(X-y))
	=-[X^k]I_{W_z\setminus\{x,y\}}\ne0.
	\]
	Thus $B$ is independent in $\mathcal H_{L_{W_z}}$.
	Consequently, the map
	\[
	\F_q\setminus T_B\longrightarrow
	\{[L]\notin\C_2:
	B\text{ is independent in }\mathcal H_L\},
	\qquad
	z\longmapsto[L_{W_z}]
	\]
	is well defined. It is injective by
	Lemma~\ref{lem:noncollision}, which proves
	\eqref{eq:failure-multiplicity}.
\end{proof}

\begin{lemma}\label{lem:bad-set-density}
	Let $U$ be a nonexpected deep hole of $\RS_q(k)$,
	where $q$ is odd and $3\le k\le(q-1)/2$.
	For $k+3\le w\le q$, let $b_w$ denote the proportion
	of bad $w$-subsets of $\F_q$. Then
	\begin{equation}\label{eq:bad-set-density}
		b_w\ge\frac wq.
	\end{equation}
\end{lemma}

\begin{proof}
	A bad $(w+1)$-subset contains at most one good
	$w$-subset. Otherwise, two such good subsets would
	intersect in $w-1\ge k+2$ points. By
	Lemma~\ref{lem:template-intersection}, their
	templates would coincide and hence cover their
	union, contradicting badness.
	
	Thus every bad $(w+1)$-subset contains at least
	$w$ bad $w$-subsets. Double-counting these
	inclusions gives
	\[
	w\,b_{w+1}\binom q{w+1}
	\le(q-w)b_w\binom qw,
	\]
	and therefore
	\[
	b_w\ge\frac w{w+1}b_{w+1},
	\qquad k+3\le w<q.
	\]
	
	Since $U$ is nonexpected, the whole field is bad:
	a rational template has a pole, while a polynomial
	template agreeing with $U$ everywhere would force
	$\deg U=k$. Hence $b_q=1$, and iteration yields
	\[
	b_w\ge
	\prod_{t=w}^{q-1}\frac{t}{t+1}\,b_q
	=\frac wq.
	\]
\end{proof}

The preceding lemmas combine the local multiplicity
of quadratic obstructions with the global density
of bad subsets.

\begin{theorem}\label{thm:compatibility-criterion}
	Let $q$ be odd, $3\le k\le(q-1)/2$, and let $U$
	be a nonexpected deep hole of $\RS_q(k)$. Then
	\begin{equation}\label{eq:compatibility-criterion}
		Q(q,k+2)\ge\frac{(q-k-2)(k+3)}q.
	\end{equation}
	If, moreover, every template agrees with $U$
	at no more than $M$ points of its domain, then
	\begin{equation}\label{eq:compatibility-strong}
		Q(q,k+2)\ge q-M.
	\end{equation}
\end{theorem}

\begin{proof}
	Put $r=k+2$. By the proof of
	Lemma~\ref{lem:failure-multiplicity}, for each
	$r$-subset $B$ and $z\notin B$,
	\[
	B\cup\{z\}\text{ is bad}
	\iff z\notin T_B.
	\]
	Every bad $(r+1)$-subset arises from exactly
	$r+1$ such pairs $(B,z)$. Double-counting gives
	\begin{align}
		\frac1{\binom qr}\sum_{|B|=r}(q-|T_B|)
		&=\frac{(r+1)\binom q{r+1}}{\binom qr}
		b_{r+1}\notag\\
		&=(q-r)b_{r+1}.
		\label{eq:failure-average}
	\end{align}
	
	By the definition of $Q$ in \eqref{eq:Q},
	\[
	Q(q,r)=\frac1{\binom qr}\sum_{|B|=r}N(B).
	\]
	Averaging \eqref{eq:failure-multiplicity} and
	applying \eqref{eq:failure-average} and
	Lemma~\ref{lem:bad-set-density}, we obtain
	\[
	Q(q,r)\ge(q-r)b_{r+1}
	\ge\frac{(q-r)(r+1)}q,
	\]
	which proves \eqref{eq:compatibility-criterion}.
	If $|T_B|\le M$ for every $B$, then
	\[
	Q(q,r)\ge
	\frac1{\binom qr}\sum_{|B|=r}(q-|T_B|)
	\ge q-M,
	\]
	proving \eqref{eq:compatibility-strong}.
\end{proof}

We apply this criterion to odd fields of order at least $41$.
The exact expression \eqref{eq:Q} admits the required uniform
bound, proved by a recurrence in $q$ in Appendix~\ref{app:quadratic}.

\begin{lemma}\label{lem:Q-cover-bound}
For every odd integer $q\ge41$, put $w_0=\ceil{q/3}+3$. Then
\begin{equation}\label{eq:uniform-Q41}
 Q(q,w_0)<9,
\end{equation}
and for $w_0\le r\le(q+3)/2$,
\begin{equation}\label{eq:uniform-density41}
 \frac{(q-r)(r+1)}q>10.
\end{equation}
\end{lemma}

\begin{theorem}\label{thm:odd-large}
Let $q\ge41$ be an odd prime power and $3\le k\le(q-1)/2$. Every deep hole of $\RS_q(k)$ has interpolation degree $k$.
\end{theorem}
\begin{proof}
For $k\le\ceil{q/3}$ we have $q>3(k-1)$, so Theorem~\ref{thm:odd-low} applies. Otherwise put $r=k+2$ and $w_0=\ceil{q/3}+3$. Then $w_0\le r\le(q+3)/2$, and Lemma~\ref{lem:Q-monotone} yields
\[
 Q(q,r)\le Q(q,w_0)<9.
\]
But a nonexpected deep hole would require
\[
 Q(q,r)\ge\frac{(q-r)(r+1)}q>10
\]
by Theorem~\ref{thm:compatibility-criterion} and Lemma~\ref{lem:Q-cover-bound}, a contradiction.
\end{proof}

\section{Completion of the Deep-Hole Classification}\label{sec:completion}
In this section, we settle the remaining odd extension-field cases
$q\in\{9,25,27\}$ and complete the proof of
Theorem~\ref{thm:main}.

For $(q,k)=(27,10)$, we use the following consequence of a
restricted-sumset theorem for finite abelian groups.

\begin{lemma}[Restricted sums and products over $\F_{27}$]\label{lem:four27}
Let $A\subseteq\F_{27}$ have size $14$. Every element of
$\F_{27}$ is a sum of four distinct elements of $A$, and also
of ten distinct elements of $A$. If $A\subseteq\F_{27}^\times$,
every nonzero field element is a product of four distinct
elements of $A$, and also of ten distinct elements of $A$.
\end{lemma}
\begin{proof}
Let $G$ be either $(\F_{27},+)$ or $(\F_{27}^\times,\cdot)$,
using additive notation for both groups in this proof. Their
orders are $27$ and $26$, and the numbers of solutions of
$2g=0$ are $1$ and $2$, respectively. Thus
$2|A|=|G|+|\{g\in G:2g=0\}|=28$ in either case.
By Guo's restricted-sumset theorem~\cite[Theorem~1]{GuoRestricted},
\[
 |A^{(2)}|\ge |G|-2,
 \qquad
 A^{(2)}=\{a+b:a,b\in A,\ a\ne b\}.
\]

For $g\in G$, let $\mu(g)$ count unordered pairs of distinct
elements of $A$ whose sum is $g$. Two distinct pairs with the
same sum cannot share an element, by cancellation. Hence
\[
 \mu(g)\le7,\qquad
 \sum_{g\in G}\mu(g)=\binom{14}{2}=91.
\]
Put $\mathcal P=\{g\in G:\mu(g)\ge3\}$. Then
\[
 91\le2(|G|-|\mathcal P|)+7|\mathcal P|,
\]
so $|\mathcal P|\ge8$. In particular,
\[
 |\mathcal P|+|A^{(2)}|\ge |G|+6>|G|.
\]

For any $g\in G$, the sets $\mathcal P$ and $g-A^{(2)}$
therefore intersect. Choose $p\in\mathcal P$ with
$g-p=a_1+a_2$ for distinct $a_1,a_2\in A$. The element $p$
has at least three pairwise disjoint two-element representations.
At most two of these meet $\{a_1,a_2\}$, so we can choose a representation
$p=a_3+a_4$ with $a_1,a_2,a_3,a_4$ all distinct.
Thus every element of $G$ is a sum of four distinct elements
of $A$.

Taking complements of four-element subsets in $A$, we obtain
\[
 \left\{\sum_{a\in T}a:T\subseteq A,\ |T|=10\right\}
 =\sum_{a\in A}a-G=G.
\]
Interpreting the group law as field addition or multiplication
proves the asserted sum and product statements.
\end{proof}

\begin{lemma}\label{lem:agreement27}
Let $U$ be the canonical representative of a nonexpected deep
hole of $\RS_{27}(10)$. Then every template agrees with $U$
on at most $13$ points of its domain.
\end{lemma}
\begin{proof}
Suppose a polynomial template $G$ of degree $10$ agrees
with $U$ on a $14$-subset $A$. Let $c\ne0$ be the leading
coefficient of $G$. Since $U$ is nonexpected, there exists
$z\notin A$ such that $G(z)\ne U(z)$. The fourteen elements
$z-x$, $x\in A$, are distinct and nonzero. By
Lemma~\ref{lem:four27}, there is a $10$-subset $T\subseteq A$
for which
\[
 P_T(z)=\prod_{x\in T}(z-x)=\frac{G(z)-U(z)}{c}.
\]
Then $G-cP_T$ has degree less than $10$ and agrees with
$U$ on the $11$ points of $T\cup\{z\}$, contrary to
\eqref{eq:deep-hole-criterion}.

Now suppose a pole template
\[
 R(X)=f(X)+\frac{c}{X-\zeta},
 \qquad \deg f<10,\quad c\ne0,
\]
agrees with $U$ on a $14$-subset
$A\subseteq\F_{27}\setminus\{\zeta\}$. Applying
Lemma~\ref{lem:four27} to the distinct nonzero elements
$(\zeta-x)^{-1}$, $x\in A$, yields a $10$-subset $T\subseteq A$
such that
\[
 \sum_{x\in T}\frac1{\zeta-x}
 =\frac{f(\zeta)-U(\zeta)}c.
\]
The expression
\[
 h(X)=f(X)+c\frac{1-P_T(X)/P_T(\zeta)}{X-\zeta}
\]
is a polynomial of degree less than $10$: the numerator
vanishes at $X=\zeta$. Since $P_T(x)=0$ for $x\in T$,
it agrees with $R=U$ on $T$. Moreover,
\[
 \begin{aligned}
 h(\zeta)
 &=f(\zeta)-c\frac{P_T'(\zeta)}{P_T(\zeta)}\\
 &=f(\zeta)-c\sum_{x\in T}\frac1{\zeta-x}
 =U(\zeta).
 \end{aligned}
\]
Thus $h$ agrees with $U$ at the $11$ distinct points of
$T\cup\{\zeta\}$, again contradicting
\eqref{eq:deep-hole-criterion}.
\end{proof}

\begin{corollary}\label{cor:2710}
Every deep hole of $\RS_{27}(10)$ has interpolation degree $10$.
\end{corollary}
\begin{proof}
Otherwise Lemma~\ref{lem:agreement27} and
\eqref{eq:compatibility-strong} imply $Q(27,12)\ge14$, whereas
\eqref{eq:Q} yields
\[
 Q(27,12)=\frac{24576}{1955}<14,
\]
a contradiction.
\end{proof}

We now combine the preceding results to treat the remaining
small extension fields.

\begin{theorem}\label{thm:small-completion}
For $q\in\{9,25,27\}$ and $3\le k\le(q-1)/2$, every
deep hole of $\RS_q(k)$ has interpolation degree $k$.
\end{theorem}
\begin{proof}
	Theorem~\ref{thm:odd-low} covers
	$k\le\ceil{q/3}$, while
	Kaipa~\cite{Kaipa} covers
	$k\ge(q-1)/2$.
	In particular, no case remains for $q=9$. The only
	remaining cases are
	\[
	(q,k)\in
	\{(25,10),(25,11),(27,10),(27,11),(27,12)\}.
	\]
	
	For $q=25$, put $r=k+2\in\{12,13\}$.
	By the monotonicity of $Q$,
	\[
	Q(25,r)\le Q(25,12)
	=\frac{42496}{7429}<6,
	\]
	whereas
	\[
	\frac{(25-r)(r+1)}{25}
	\ge\frac{168}{25}>6.
	\]
	This contradicts
	\eqref{eq:compatibility-criterion}
	for a nonexpected deep hole.
	
	For $q=27$, Corollary~\ref{cor:2710}
	settles $k=10$.
	For $k\in\{11,12\}$, put
	$r=k+2\in\{13,14\}$. Then
	\[
	Q(27,r)\le Q(27,13)
	=\frac{142848}{37145}<4,
	\]
	while
	\[
	\frac{(27-r)(r+1)}{27}
	\ge\frac{65}{9}>4.
	\]
	Again \eqref{eq:compatibility-criterion}
	excludes nonexpected deep holes.
\end{proof}

\begin{proof}[Proof of Theorem~\ref{thm:main}]
The root bound shows that every degree-$k$ canonical polynomial
represents a deep hole. Conversely, Theorem~\ref{thm:even} covers
even $q$, and Theorem~\ref{thm:odd-large} covers odd $q\ge41$.
The odd prime fields are settled by Zhuang, Cheng, and
Li~\cite{ZCL}. The only odd proper extension fields below $41$
are $9,25,27$, which are covered by
Theorem~\ref{thm:small-completion}.
\end{proof}

\section{Concluding Remarks}
In this paper, we have classified the deep holes of full-length Reed--Solomon codes
for $3\le k\le\lfloor q/2\rfloor$ over every finite field. Together
with the earlier high-dimensional classifications, this proves the
corrected Cheng--Murray conjecture.
The proof uses the full evaluation set in two complementary ways.
The marked point determines a boundary section of the tangent
envelope, and its first transverse coefficient constrains the
interpolation polynomial. Beyond the range of these differential
constraints, overlapping local syndromes still impose restrictions:
distinct failed template extensions correspond to distinct quadratic
obstructions. Compatibility bounds and exact matching counts rule out nonexpected deep holes in the remaining range, while a restricted sum--product argument handles the last small-field case.

A natural next step concerns the primitive Reed--Solomon code
$\RS(\F_q^\times,k)$, evaluated at the nonzero field elements.
For $2\le k\le q-3$, Wu and Hong~\cite{WuHong} showed that,
besides the expected degree-$k$ deep holes, the words
represented by $aX^{q-2}+f(X)$, where $a\in\F_q^\times$
and $\deg f<k$, are also deep holes. They conjectured
that these two families are exhaustive when $q$ is odd.
The analogous classification fails in even characteristic:
for $q\ge8$, restricting the $k=2$ example
$U(X)=X^{q/2}$ from Section~\ref{sec:interpolation}
to $\F_q^\times$ gives an additional deep hole, while
for $k=q-4$ the words represented by $aX^{q-3}+f(X)$
form another family~\cite{ZFL,Kaipa}. A complete classification
in all characteristics remains a natural problem.
The tangent-envelope and local-syndrome methods developed
here may offer a route toward its resolution.

More generally, let $D\subseteq\F_q$ have size $n$, and write
$\RS(D,k)$ for the corresponding evaluation code.
By the equivalence between deep holes and one-coordinate
MDS extensions of the dual code~\cite[Proposition~1]{Kaipa},
the MDS-extension theorem of Seroussi and Roth~\cite{SR}
implies that, if $2\le n-k$ and
$k\ge\lfloor(q-1)/2\rfloor$, excluding $n-k=3$
when $q$ is even, every deep hole of $\RS(D,k)$
is represented on $D$ by one of the functions
\[
aX^k+f(X),\qquad \frac{a}{X-\beta}+f(X),
\]
where $a\in\F_q^\times$, $\deg f<k$, and
$\beta\in\F_q\setminus D$ in the second case.
Classifying deep holes for general $D$ when
$k<\lfloor(q-1)/2\rfloor$ remains an interesting problem,
for which the structure of the evaluation set may play
an essential role.

Another natural problem is to classify the deep holes of
projective Reed--Solomon codes in general dimensions.
These codes have length $q+1$, with evaluation at infinity
defined as the coefficient of $X^{k-1}$. Their covering-radius
and deep-hole problems are closely related to normal
rational curves. Although the deep holes are classified
for redundancy at most four~\cite{Kaipa,ZhangWanKaipaPRS},
a complete classification in general dimensions remains open.
It would be interesting to investigate whether the
geometric and local-syndrome methods used here extend to
the projective setting.

\appendix
\section{A Uniform Bound for Quadratic Matching Counts}\label{app:quadratic}
\begin{proof}[Proof of Lemma~\ref{lem:Q-cover-bound}]
Let $q\ge41$ be an odd integer and put
$w_0(q)=\ceil{q/3}+3$. We first derive a closed form for
$Q(q,w)$ that is suitable for comparing successive values of $q$.

For $i\in\{1,3\}$, the graph counted by $A_i(q,w)$ has
$i$ isolated vertices and $(q-i)/2$ disjoint edges.
An isolated vertex contributes a factor $1+t$, whereas
an edge contributes $1+2t$ to the independence polynomial.
Thus, by \eqref{eq:Ai},
\[
A_i(q,w)=[t^w](1+t)^i(1+2t)^{(q-i)/2},
\]
where $[t^w]$ denotes the coefficient of $t^w$.
Setting $h=(q-3)/2$, and using
$N_1=q(q+1)/2$ and $N_3=q(q-1)/2$ in
\eqref{eq:Q}, we obtain
\begin{equation}\label{eq:Q-generating}
\begin{aligned}
Q(q,w)=\frac{q}{2\binom qw}[t^w](1+2t)^h
\bigl(2q+6qt+(5q-1)t^2+(q-1)t^3\bigr).
\end{aligned}
\end{equation}

Assume that
$3\le w\le(q+3)/2$ and put $d=h-w+3$.
By the binomial theorem, $[t^m](1+2t)^h=2^m\binom hm,$
where $\binom hm=0$ for $m>h$. Hence
\[
\begin{aligned}
&[t^w](1+2t)^h
 \bigl(2q+6qt+(5q-1)t^2+(q-1)t^3\bigr)\\
=&2q[t^w](1+2t)^h+6q[t^{w-1}](1+2t)^h +(5q-1)[t^{w-2}](1+2t)^h
 +(q-1)[t^{w-3}](1+2t)^h\\
=&2q\,2^w\binom hw+6q\,2^{w-1}\binom h{w-1} +(5q-1)2^{w-2}\binom h{w-2}
 +(q-1)2^{w-3}\binom h{w-3}\\
=&2^{w-3}\biggl(16q\binom hw+24q\binom h{w-1}+2(5q-1)\binom h{w-2}
 +(q-1)\binom h{w-3}\biggr).
\end{aligned}
\]
The four binomial coefficients can be expressed relative
to $\binom h{w-3}$ using
\[\frac{\binom h{w-2}}{\binom h{w-3}}
 =\frac{d}{w-2}, \qquad
\frac{\binom h{w-1}}{\binom h{w-3}}
 =\frac{d(d-1)}{(w-1)(w-2)}, \qquad
\frac{\binom hw}{\binom h{w-3}}
 =\frac{d(d-1)(d-2)}{w(w-1)(w-2)}.
\]
Upon substitution, collecting terms over the denominator
$w(w-1)(w-2)$ and using $d=(q-2w+3)/2$
yields the factor
$(q-1)(q-w+1)P(q,w)$, where
\[
P(q,w)=2q^2-4qw+6q+w^2-w.
\]
Consequently,
\begin{equation}\label{eq:Q-closed}
Q(q,w)=
\frac{2^{w-4}q(q-1)(q-w+1)P(q,w)}{w(w-1)(w-2)}
\frac{\binom{(q-3)/2}{w-3}}{\binom qw}.
\end{equation}

The choice of the shift $(q,w)\mapsto(q+6,w+2)$
is dictated by $w_0(q+6)=w_0(q)+2$.
For this shift, the required binomial ratios are
\[
\begin{aligned}
\frac{\binom{h+3}{w-1}}{\binom h{w-3}}
 &=\frac{(h+1)(h+2)(h+3)}
 {(w-1)(w-2)(h-w+4)},\\
\frac{\binom qw}{\binom{q+6}{w+2}}
 &=\frac{(w+1)(w+2)\prod_{i=1}^4(q-w+i)}
 {\prod_{i=1}^6(q+i)}.
\end{aligned}
\]
Substituting these into the ratio of \eqref{eq:Q-closed},
 we obtain
\begin{equation}\label{eq:Q-recurrence}
\frac{Q(q+6,w+2)}{Q(q,w)}=
\frac{\prod_{i=2}^{5}(q-w+i)}
     {q(q+2)(q+4)(q-2w+5)}
\frac{P(q+6,w+2)}{P(q,w)}.
\end{equation}
All denominators are positive in the parameter range
used below.

We now prove \eqref{eq:uniform-Q41}.
Since $q$ is odd, write $q=6n+j$ with
$j\in\{1,3,5\}$, and define
\[
R_j(n)=\frac{Q(6n+j+6,w_0(6n+j)+2)}
              {Q(6n+j,w_0(6n+j))}.
\]
For $j=1,3,5$, the corresponding values of $w_0(6n+j)$
are $2n+4,2n+4,2n+5$, respectively.
Substitution into \eqref{eq:Q-recurrence} gives
\begin{align*}
R_1(n)&=\frac{4n(4n-1)(4n+1)}
 {3(n-1)(6n+1)(6n+5)}
 \frac{14n^2+13n+1}{14n^2-15n+2},\\
R_3(n)&=\frac{4(n+1)^2(4n+1)(4n+3)}
 {3n^2(6n+5)(6n+7)}
 \frac{14n+15}{14n+1},\\
R_5(n)&=\frac{4(n+1)^2(4n+3)(4n+5)}
 {3n^2(6n+5)(6n+7)}.
\end{align*}
We show that each ratio is less than one in its
admissible range.

For $R_1$, we have $n\ge7$ and
\[
\frac{n}{n-1}\le\frac76,\qquad
\frac{14n^2+13n+1}{14n^2-15n+2}\le\frac75.
\]
Since $(4n-1)/(6n+1)$ and $(4n+1)/(6n+5)$
are both less than $2/3$, we obtain
\[
R_1(n)<\frac43\cdot\frac76\cdot
\left(\frac23\right)^2\cdot\frac75
=\frac{392}{405}<1.
\]
For $R_3$, again $n\ge7$, and
\[
\frac{n+1}{n}\le\frac87,\qquad
\frac{14n+15}{14n+1}<\frac{n+1}{n}.
\]
 The two remaining linear ratios
are less than $2/3$. Hence
\[
R_3(n)<\frac43\left(\frac87\right)^3
\left(\frac23\right)^2
=\frac{8192}{9261}<1.
\]
Finally, $R_5$ has $n\ge6$. Then
\[
\frac{n+1}{n}\le\frac76,\qquad
\frac{4n+3}{6n+5}<\frac23,\qquad
\frac{4n+5}{6n+7}<\frac34.
\]
Hence
\[
R_5(n)<\frac43\left(\frac76\right)^2
\cdot\frac23\cdot\frac34
=\frac{49}{54}<1.
\]
Thus $Q(q,w_0(q))$ decreases when $q$ is increased
by six within any odd residue class modulo six.
The first admissible values are $q=43,45,41$ for
$j=1,3,5$, respectively, and by \eqref{eq:Q-closed}, we have
\[
\begin{aligned}
Q(43,18)&=\frac{19103744}{4091349}<5,\\
Q(45,18)&=\frac{476872704}{58642669}<9,\\
Q(41,17)&=\frac{229376}{33263}<7.
\end{aligned}
\]
It follows that $Q(q,w_0(q))<9$ for every odd
$q\ge41$, proving \eqref{eq:uniform-Q41}.

It remains to prove \eqref{eq:uniform-density41}.
For fixed $q$, set
\[
T(r)=\frac{(q-r)(r+1)}q.
\]
Since $T''(r)=-2/q<0$, the minimum of $T$ on
$[w_0(q),(q+3)/2]$ occurs at an endpoint.
Moreover, $T'(r)=(q-1-2r)/q$, so $T$ is increasing
for $r\le(q-1)/2$. As
\[
\frac q3+3\le w_0(q)<\frac q3+4
\le\frac{q-1}{2}\qquad(q\ge41),
\]
we have
\[
\begin{aligned}
T(w_0(q))
&\ge T\left(\frac q3+3\right)
 =\frac{2q}{9}+\frac53-\frac{12}{q}
 \ge\frac{3869}{369}>10,\\
T\left(\frac{q+3}{2}\right)
&=\frac q4+\frac12-\frac{15}{4q}
 \ge\frac{437}{41}>10.
\end{aligned}
\]
Both expressions used in the lower bounds are increasing
in $q$, and their displayed constants are obtained at
$q=41$. Therefore $T(r)>10$ throughout the interval,
which proves \eqref{eq:uniform-density41}.
\end{proof}

\section*{Acknowledgments}
This work was supported in part by the National Key Research and Development
Program of China under Grant 2022YFA1004900, in part by the National
Natural Science Foundation of China under Grant 62571301, and in part by the Shandong Provincial Natural Science Foundation under Grant ZR2025QA05.

\section*{AI Declaration}
During the preparation of this manuscript, the author used
OpenAI's ChatGPT (GPT-6) to assist with
exploratory discussions, checking some of the main
proofs, and improving the exposition. The author independently
verified the mathematical content and references and
takes full responsibility for the manuscript.

\end{document}